\documentclass[11pt,letter]{article}

\usepackage{amsmath}
\usepackage{amssymb}
\usepackage{amsthm}
\usepackage{epsfig}
\usepackage{multicol}

\usepackage{algorithmic}
\usepackage{algorithm}
\usepackage{verbatim}

 \newtheorem{prop}{Proposition}
 \newtheorem{theorem}[prop]{Theorem}

\theoremstyle{definition}
 \newtheorem{definition}[prop]{Definition}

\newcommand{\pref}{\succ}

\def\calG{\mathcal{G}}
\def\calV{\mathcal{V}}
\def\calE{\mathcal{E}}

\def\calR{\mathcal{R}}

\def\vect{\mathbf{t}}
\def\vecs{\mathbf{s}}
\def\vecr{\mathbf{r}}

\def\vecq{\mathbf{q}}

\def\eps{\varepsilon}

\newcommand{\shift}{\ensuremath{\mathit{shf}}}
\newcommand{\push}{\ensuremath{\mathit{psh}}}

\newcommand{\rank}{\ensuremath{\mathit{rank}}}
\newcommand{\succc}{\ensuremath{\mathit{succ}}}
\newcommand{\precc}{\ensuremath{\mathit{prec}}}

\newcommand{\integers}{\mathbb{Z}}

\newcommand{\p}{{\rm P}}
\newcommand{\np}{{\rm NP}}

\newcommand{\opt}{{\ensuremath{\mathrm{opt}}}}

\newcommand{\Z}{\mathbb Z}

\begin{document}
\title{Campaign Management under Approval-Driven Voting Rules}
\author{Ildik\'o Schlotter\\
  Budapest University of Technology\\ and Economics,
  Hungary
\and
  Piotr Faliszewski \\
  AGH University\\
  Krak\'ow, Poland
\and
  Edith Elkind\\
  Department of Computer Science\\
  University of Oxford\\
  United Kingdom
}

\maketitle

\begin{abstract}
  Approval-like voting rules, such as Sincere-Strategy
  Pre\-fe\-rence-Based Approval voting (SP-AV), the Bucklin rule (an
  adaptive variant of $k$-Approval voting), and the Fallback rule (an
  adaptive variant of SP-AV) have many desirable properties: for
  example, they are easy to understand and encourage the candidates to
  choose electoral platforms that have a broad appeal. In this paper,
  we investigate both classic and parameterized computational
  complexity of electoral campaign management under such rules. We
  focus on two methods that can be used to promote a given candidate:
  asking voters to move this candidate upwards in their preference
  order or asking them to change the number of candidates they approve
  of.  We show that finding an optimal campaign management strategy of
  the first type is easy for both Bucklin and Fallback.  In contrast,
  the second method is computationally hard even if the degree to
  which we need to affect the votes is small.  Nevertheless, we identify a
  large class of scenarios that admit fixed-parameter
  tractable algorithms.
\end{abstract}

\section{Introduction}\label{sec:intro}

Approval voting---a voting rule that asks each voter to report which
candidates she approves of and outputs the candidates with the largest
number of approvals---is one of the very few election systems that
have a real chance of replacing Plurality voting in political
elections. It has many attractive theoretical properties, 
and its practical usefulness is supported by the experimental
results of Laslier and Van der
Straeten~\cite{las-str:j:approval-experiment} and Van der Straeten
et al.~\cite{bla-las-str-sau:j:strategic-sincere-experiment}.
Some professional organizations,
such as, e.g., the Mathematical Association of America (MAA), 
Institute for Operations Research and Management Sciences (INFORMS),
or  Institute of Electrical and Electronics Engineers (IEEE),
employ Approval voting to select their leaders.
One of the major attractions of Approval voting is that, in
contrast to the more standard Plurality voting, under Approval voting
the candidates can benefit from running their campaigns in a
consensus-building fashion, i.e., by choosing a platform that appeals
to a large number of voters.

Nonetheless, Approval voting has certain disadvantages as well. 
Perhaps the most significant of them is its limited expressivity.
Indeed, even a voter that approves of several candidates may like 
some of them more than others; however, Approval voting does not 
allow her to express this. Therefore, it is desirable to have a voting rule
that operates similarly to Approval, yet takes voters' preference orders 
into account.

Several such voting rules have been proposed. For instance, 
the Bucklin rule (also known as the majoritarian compromise) 
asks the voters to gradually increase the number
of candidates they approve of, until some candidate is approved
by a majority of the voters. The winners are the candidates that receive the largest number
of approvals at this point. 
In a simplified version of this rule, which is popular in the computational social choice
literature~\cite{con-pro-ros-xia-zuc:c:unweighted-manipulation,con-xia:j:possible-necessary-winners,elk-fal-sli:c:votewise-dr}, 
the winners are all candidates that are approved by a majority of 
the voters in the last round.
Under both variants of the Bucklin rule, the common approval threshold is lowered gradually, 
thus reflecting the voters' preferences.
However, this common threshold may move past 
an individual voter's personal approval threshold,  
forcing this voter to grant approval to a candidate that she does not approve of.
To alleviate this problem, Brams and Sanver~\cite{bra-san:j:fv-pav} 
have recently introduced a new election system, which they call Fallback
voting. This system works similarly to the Bucklin rule, 
but allows each voter to only approve of a limited number of candidates;
its simplified version can be defined similarly to the simplified
Bucklin voting.

With variants of Approval voting gaining wider acceptance, it becomes
important to understand whether various activities associated with running
an approval-based electoral campaign are computationally tractable.
Such activities can be roughly classified into
benign, such as winner determination, and malicious, such as 
manipulation and control; an ideal voting rule admits
polynomial-time algorithms for the benign activities, but not for the malicious ones. 
However, there is an election-related activity that defies such classification, 
namely, bribery, or campaign management
\cite{fal-hem-hem:j:bribery,elk-fal-sli:c:swap-bribery,elk-fal:c:shift-bribery,bau-fal-lan-rot:c:lazy-voters,bre-che-fal-nic-nie:c:shift-bribery-fpt}. 
Both of these terms are used
for actions that aim to make a given candidate an election winner by means
of spending money on individual voters so as to change their preference rankings;
these actions can be benign if the money is spent on legitimate activities, 
such as advertising, or malicious, if the voters are paid to vote non-truthfully.

Now, winner determination for all approval-based rules listed above is clearly
easy, and the complexity of manipulation and especially control under such
rules is well understood
\cite{hem-hem-rot:j:destructive-control,bau-erd-hem-hem-rot:b:computational-apects-of-approval-voting,erd-fel-rot-sch:j:fallback-buclin-control-theory,erd-fel-rot-sch:j:fallback-buclin-control-experiments,fal-hem-hem:c:weighted-control,fal-rei-rot-sch:j:manipulation-bribery-bucklin-fallback}.
Thus, in this paper we focus on algorithmic aspects of electoral
campaign management. Following Elkind, Faliszewski, and Slinko
\cite{elk-fal-sli:c:swap-bribery} and Elkind and Faliszewski~\cite{elk-fal:c:shift-bribery}
(see also \cite{dor-sch:j:parameterized-swap-bribery,bre-che-fal-nic-nie:c:shift-bribery-fpt}) 
who study this problem for a variety
of preference-based voting rules, we model the campaign management
setting using the framework of {\em shift bribery}.
Under this framework, each voter $v$ is associated with a cost function $\pi$, 
which indicates, for each $k>0$, how much it would cost
to convince $v$ to promote the target candidate $p$ by $k$ positions in her
vote. The briber (campaign manager) wants to make $p$ a winner by spending
as little as possible.
This framework can be used to model a wide variety
of campaign management activities, ranging from one-on-one meetings to
phone-a-thons to direct mailing, each of which has a per-voter cost that
may vary from one voter to another.

Note, however, that in the context of approval-based voting rules, we
can campaign in favor of a candidate $p$ even without changing the
preference order of any voter. Specifically, if some voter $v$ ranks
$p$ in position $k$ and currently approves of $k-1$ candidates, we can
try to convince $v$ to lower her approval threshold so that she
approves of $p$ as well.  Similarly, we can try to convince a voter to
be more stringent and withdraw her approval from her least preferred
approved candidate; this may be useful if that candidate is $p$'s
direct competitor. Arguably, a voter may be more willing to change her
approval threshold than to alter her ranking of the
candidates. Therefore, such campaign management tactics may be within
the campaign manager's budget, even when she cannot afford the more
direct approach discussed in the previous paragraph.  We will refer to
this campaign management technique as ``support bribery''; a variant
of this model has been considered by Elkind, Faliszewski, and
Slinko~\cite{elk-fal-sli:c:swap-bribery} (a similar, but somewhat
different, variant of this model was studied by Baumeister et
al.~\cite{bau-fal-lan-rot:c:lazy-voters}).

In this paper, we investigate both campaign
management activities discussed above, i.e., shift bribery and 
support bribery, from the algorithmic perspective. 
 We consider five approval-based voting rules,
namely, SP-AV (as formalized by Brams and
Sanver~\cite{bra-san:j:critical-strategies-under-approval}), Bucklin
(both classic and simplified), and Fallback (both classic and
simplified).  We show that shift bribery is easy with respect to both
variants of the Bucklin rule, as well as both variants of the Fallback
rule. The argument for the simplified version of both rules relies on
dynamic programming, while for the classic version of these rules we
use a more involved flow-based approach.  In contrast, support bribery
tends to be hard; this holds even if we parameterize this problem by
the number of voters to be bribed or the total change in the approval
counts, and use very simple bribery cost functions.  Nevertheless, we
identify a natural class of bribery cost functions for which support
bribery is fixed-parameter tractable (FPT) with respect to the latter
parameter.  Interestingly, some of our hardness results hold even for
the case of single-peaked profiles, where one often---though certainly
not always---expects
tractability~\cite{wal:c:uncertainty-in-preference-elicitation-aggregation,con:j:eliciting-singlepeaked,fal-hem-hem-rot:j:single-peaked-preferences,bra-bri-hem-hem:c:sp2,fal-hem-hem:j:nearly-sp}.
On the other hand, some of the problems considered in this paper 
do become easy when the input election can be assumed to be single-peaked: in particular, 
we describe a good (FPT) approximation algorithm for support bribery under SP-AV for the
single-peaked domain.

The rest of this paper is organized as follows. In the next section we
formally define our model of elections and the voting systems we study,
as well as provide the necessary background on (parameterized) computational
complexity. We then present our algorithms for shift bribery (Section~\ref{sec:shift}), followed
by hardness results (classic and parameterized) and FPT algorithms for
support bribery (Section~\ref{sec:support}). Section~\ref{sec:SP} contains our results on
support bribery for single-peaked elections. We conclude the paper by
presenting directions for future research.

\section{Preliminaries}
An {\em election} is a pair $E = (C,V)$, where $C = \{c_1, \ldots, c_m\}$
is the set of {\em candidates} and $V = (v^1, \ldots, v^n)$ is the list of {\em voters}.
Each voter $v^i$ is associated with a {\em preference order} $\pref^i$,
which is a total order over $C$, and an {\em approval count} $\ell^i\in[0, |C|]$;
voter $v^i$ is said to {\em approve} of the top $\ell^i$ candidates 
in her preference order. We denote by $\rank(c,v)$ the position
of candidate $c$ in the preference order of voter $v$: $v$'s most preferred  
candidate has rank $1$ and her least preferred candidate has rank $|C|$.
A {\em voting rule} is a mapping that given an election $E = (C,V)$ 
outputs a set $W \subseteq C$ of {\em election winners}.

We say that an election $(C,V)$ is {\em single-peaked} if there is an
order $\triangleright$ of the candidates (called the {\em societal
  axis}) such that each voter's preference order $\pref^i$ satisfies
the following condition: for every triple of candidates $(a, b, c)$,
if $a \mathrel\triangleright b \mathrel\triangleright c$ or $c
\mathrel\triangleright b \mathrel\triangleright a$, then $a \pref^i b$
implies $b \pref^i c$.  Equivalently, an election is single-peaked if
there is an order $\triangleright$ over the set of candidates such
that for each prefix of each vote, the set of candidates included in
this prefix forms an interval with respect to $\triangleright$.  Given
an election, it is easy to verify if it is single-peaked and, if so,
to compute one of the societal axes for it in polynomial
time~\cite{bar-tri:j:stable-matching-from-psychological-model,esc-lan-ozt:c:single-peaked-consistency}. (Interestingly,
deciding if a profile is in some sense close to being single-peaked is
typically an $\np$-hard
task~\cite{erd-lac-pfa:c:nearly-sp,bre-che-woe:c:nearly-restricted-domain},
unless we know the axis with respect to which we measure the
closeness~\cite{fal-hem-hem:j:nearly-sp}).  Intuitively, the notion of
single-peakedness captures scenarios where the electorate is focused
on a single one-dimensional issue such as, e.g., the left-to-right
political spectrum or the military spending.

\medskip

\noindent{\bf Voting Rules\ }
We now describe the voting rules that will be considered in this paper.
In what follows, we denote the number of voters by $n$.
Under {\em $k$-Approval} each candidate gets one point from each voter
that ranks her in top $k$ positions. The $k$-Approval score $s_k(c)$
of a candidate $c\in C$ is the total number of points that she gets,
and the winners are the candidates with the highest score.  The
\emph{Bucklin rule}, which can be thought of as an adaptive version of
$k$-Approval, is defined as follows. Given a candidate $c\in C$, let
$s_B(c)$ denote the smallest value of $k$ such that at least
$\lfloor\frac{n}{2}\rfloor+1$ voters rank $c$ in the top $k$
positions; we say that $c$ {\em
  wins in round $s_B(c)$}.  The quantity $k_B=\min_{c\in C}s_B(c)$ is
called the {\em Bucklin winning round}.  Observe that no candidate
wins in any of the rounds $\ell<k_B$ and at least one candidate wins
in round $k_B$. The Bucklin winners are the candidates with the
highest $k_B$-Approval score.  Under the simplified Bucklin rule, the
winners are the candidates whose $k_B$-Approval score is at least
$\lfloor\frac{n}{2}\rfloor+1$; all Bucklin winners are simplified
Bucklin winners, but the converse is not necessarily true.

We observe that $k$-Approval, despite its name, ignores the approval
counts entirely: a candidate $c$ may fail to get a point from a voter
$v^i$ who approves of her (if $\ell^i\ge \rank(c, v^i) >k$), or obtain
a point from a voter $v^j$ who does not approve of her (if $\ell^j <
\rank(c, v^j) \le k$).  Similarly, neither version of the Bucklin rule
uses the information provided by the approval counts.  In contrast,
the {\em SP-AV rule} \cite{bra-san:j:critical-strategies-under-approval}
relies heavily on the approval counts: we define a candidate's {\em approval score}
to be the number of voters that approve of her, and the winners under SP-AV
are the candidates with the highest approval score.  Finally, {\em Fallback
  voting} \cite{bra-san:j:fv-pav} makes use of both the preference
orders and the approval counts. Specifically, under this rule we apply
the Bucklin rule to the election obtained by deleting each voter's
non-approved candidates from her preference ranking. Since the
preference orders are truncated, it may happen that no candidate is
ranked by more than half of the voters, in which case the candidates
with the highest approval score are elected. We can replace
the Bucklin rule with the simplified Bucklin rule in this
construction; we will refer to the resulting rule as the {\em
  simplified Fallback rule}.

\medskip

\noindent{\bf Parameterized Complexity\ }
The framework of parameterized complexity deals with computationally
hard problems. In a parameterized problem, each input instance $I$
includes an integer $k$ called the \emph{parameter}, and
the aim is to design algorithms that are efficient if the value of the
parameter is small. Formally, a problem is said to be
\emph{fixed-parameter tractable (FPT)} with respect to parameter $k$
if it admits an algorithm whose running time on input $(I,k)$ is
$f(k)|I|^{O(1)}$ for some computable function $f$, where $|I|$ is the description size of $I$; note that the
exponent of $|I|$ does not depend on $k$.  Though $f$ is typically an
exponential function, such an algorithm is usually more efficient
than, for example, one that runs
in time $\Theta(|I|^{k})$.

Parameterized complexity also has a hardness theory, which relies on
\emph{parameterized reductions}.  Given two parameterized problems $Q$
and $Q'$, we say that $Q$ {\em reduces} to $Q'$ if there is an FPT-computable
function $f$ such that for each input $(x,k)$ it holds that $(x,k) \in Q$ if and
only if $f(x,k)=(x',k') \in Q'$ and,
moreover, $k' = g(k)$ for some function $g$. That is, the parameter of
the transformed instance only depends on the parameter of the original
instance. We call $f$ a {\em parameterized reduction} from $Q$ to
$Q'$.

An analog of the class NP in the parameterized hierarchy is W[1]: a
parameterized problem is in W[1] if it admits a parameterized
reduction to the problem of deciding whether a given Turing machine
accepts a given input word in at most $k$ steps.  The class W[2] is
the next class in the parameterized hierarchy, and we have
FPT $\subseteq$ W[1] $\subseteq$ W[2].  A problem is said to be
W[1]-{\em hard} (respectively, W[2]-{\em hard}), if all problems in
W[1] (respectively, W[2]) can be reduced to it by a parameterized
reduction.  It is conjectured that FPT $\neq$ W[1]. Just as
NP-hardness of a problem indicates that this problem is unlikely to be
polynomial-time solvable, a W[1]-hardness (or, worse yet, W[2]-hardness) result
means that the problem (with the given parameterization) is unlikely to admit an FPT algorithm.

To prove W[1]-hardness (or W[2]-hardness) of a parameterized problem
$Q$, it suffices to show a parameterized reduction from some
parameterized problem already known to be W[1]-hard (respectively, W[2]-hard).
In our hardness proofs, we will use the W[1]-hard
\textsc{multicolored
  clique} problem~\cite{fel-her-ros-via:j:multicolored-hardness} and
the W[2]-hard \textsc{dominating set} problem~\cite{dow-fel:b:parameterized}.
\begin{definition}
  In the \textsc{multicolored clique} problem we are given a graph
  $G=(\calV,\calE)$, an integer $k$, and a partition of the vertex set
  $\calV$ into $k$ independent sets $\calV^1, \dots, \calV^k$. We ask
  if $G$ contains a $k$-clique. We take $k$ to be the parameter.
\end{definition}
\begin{definition}
  In the \textsc{dominating set} problem we are given a graph $G$ and
  a positive integer $k$. We ask if $G$ has a dominating set of size
  at most $k$, that is, if there exists a subset $S$ of $G$'s vertices such
  that (a) $|S| \leq k$, and (b) each vertex not in $S$ has a neighbor
  in $S$. We take $k$ to be the parameter.
\end{definition}
For a more extensive treatment of parameterized complexity, we refer
the reader to the several excellent textbooks on this subject
\cite{dow-fel:b:parameterized,nie:b:invitation-fpt,flu-gro:b:parameterized-complexity}.

\medskip

\noindent{\bf Campaign Management\ }
The following definition is adapted
from the work of Elkind and Faliszewski~\cite{elk-fal:c:shift-bribery}, which
builds on the ideas of Elkind, Faliszewski, and Slinko~\cite{elk-fal-sli:c:swap-bribery}.

\begin{definition}
  Let $\calR$ be a voting rule. An instance of {\sc $\calR$-shift
    bribery} problem is a tuple $I = (C, V,\Pi, p)$, where 
  $C = \{p,c_1, \ldots, c_{m-1}\}$ is a set of candidates, 
  $V = (v^1, \ldots, v^n)$ is a list of voters together with their preference
  orders over $C$ (and approval counts, if $\calR$ uses them), 
  $\Pi = (\pi^1, \ldots, \pi^n)$ is a family of cost functions, where each
  $\pi^i$ is a non-decreasing function from $[0, |C|]$ to $\mathbb
  Z^+\cup\{+\infty\}$ that satisfies $\pi^i(0)=0$ 
  (each function $\pi^i$ is specified by listing its values at $0, \dots, |C|$), and $p\in C$ is a
  designated candidate.
  The goal is to find 
  a vector $\vect = (t_1, \ldots, t_n) \in (\Z^+)^n$ with the following properties: 
  (a) if for each $i=1,
  \dots, n$ we shift $p$ upwards in the $i$-th vote by $t_i$
  positions, then $p$ becomes an $\calR$-winner of the resulting
  election, and
  (b) for all $\vecs = (s_1, \ldots, s_n) \in (\Z^+)^n$
  that satisfy condition~(a) it holds that
  $\sum_{i=1}^{n}\pi^i(t_i)\le \sum_{i=1}^{n}\pi^i(s_i)$.
  We set $\opt(I) = \sum_{i=1}^{n}\pi^i(t_i)$.
\end{definition}
In words, $\pi^i(k)$ is the cost of shifting the preferred
candidate $p$ upwards by $k$ positions in the preferences of the $i$-th voter.
We will refer to the vector $\vect = (t_1,\ldots, t_n)$ as a {\em shift action}
or simply a {\em (shift) bribery}, 
and denote by $\shift(C,V,\vect)$ (or by $\shift(E, \vect)$) the election obtained from $E=(C,V)$ by
shifting $p$ upwards by $t_i$ positions in the $i$-th vote, for each $i=1, \dots, n$.
If $\rank(p, v^i)=k$, but a shift action prescribes shifting $p$
by $k' \geq k$ positions in $v^i$, we simply place $p$
on top of $v^i$.
Also, we write $\Pi(\vect)=\sum_{i=1}^n\pi^i(t_i)$ to denote the cost of
a shift action $\vect$.

We say that a shift action $\vect=(t_1, \dots, t_n)$ is {\em minimal}
for $I$, if $\Pi(\vect)=\opt(I)$, $p$ is a winner in $\shift(C, V,
\vect)$, but for every shift action $\vecs\neq\vect$ such that $s_i\le t_i$ for all $i=1, \dots, n$ it holds that
$p$ is not a winner in $\shift(C, V, \vecs)$. Note that
an optimal shift action is not necessarily minimal, as it may include
some shifts of cost zero that are not needed to make $p$ a winner.

Shift bribery does not change the voters' approval counts.
A more general notion of bribery, which is relevant for SP-AV 
and (simplified) Fallback voting, was proposed
by Elkind, Faliszewski, and Slinko~\cite{elk-fal-sli:c:swap-bribery}
in the technical report version of their paper.
Specifically, they
defined \emph{mixed bribery} for SP-AV, where the briber can both shift the
preferred candidate and change the voters' approval counts.
In this work, we find it more convenient
to separate these two types of bribery. Thus, we will now define
\emph{support bribery}, which focuses on changing the number of
approved candidates. 

To define support bribery, we need to be able to specify the costs of
increasing/decreasing approval counts for the voters.  Formally, we
assume that each voter $v^i$ has a {\em support bribery cost function}
$\sigma^i \colon \mathbb Z \rightarrow \mathbb Z^+\cup\{+\infty\}$,
which satisfies (a) $\sigma^i(0) = 0$, and (b) for each $k>0$ it holds that $\sigma^i(k)
\leq \sigma^i(k+1)$ and $\sigma^i(-k) \leq \sigma^i(-k-1)$.
For a given $k\in \mathbb Z$, 
we interpret $\sigma^i(k)$ as the cost of convincing $v^i$ to approve
of $\ell^i+k$ candidates. Clearly, it suffices to define $\sigma^i$
on $[-\ell^i, |C|-\ell^i]$, where $\ell^i$ is the approval count 
of $v^i$. %

\begin{definition}
  Let $\calR$ be a voting rule. An instance of $\calR$-\textsc{support
    bribery} problem is a tuple $I = (C,V,\Sigma,p)$, where $C = \{p, 
  c_1,\ldots, c_{m-1}\}$ is a set of candidates, $V = (v^1, \ldots, 
  v^n)$ is a list of voters, where each voter $v^i$ is represented by 
  her preference order $\pref^i$ and her approval count $\ell^i$, and 
  $\Sigma = (\sigma^1, \ldots, \sigma^n)$ is a family of support 
  bribery cost functions (each function $\sigma^i$ is represented by listing its values on
  $[-\ell^i, |C|-\ell^i]$). The goal is to find 
  a vector $\vect = (t_1, \ldots, t_n) \in   \integers^n$ with the following properties: 
  (a) if for each $i=1, \dots, n$ voter $v^i$
  changes her approval count from $\ell^i$ to $\ell^i+t_i$, then $p$
  is an $\calR$-winner of the resulting election, and
  (b) for all $\vecs = (s_1, \ldots, s_n) \in \Z^n$
  that satisfy condition~(a) it holds that
  $\sum_{i=1}^{n}\sigma^i(t_i)\le \sum_{i=1}^{n}\sigma^i(s_i)$.
  We set $\opt(I) = \sum_{i=1}^{n}\sigma^i(t_i)$.
\end{definition}

We refer to the vector $\vect = (t_1,\ldots, t_n)$ as a {\em push
  action} or a {\em (support) bribery}.  We denote by
$\push(C,V,\vect)$ (or by $\push(E,\vect)$) the election obtained
from election $E=(C,V)$ by setting, for each $i$, the approval count
of the $i$-th voter to be $\ell^i+t_i$.  If $\ell^i+t_i<0$ or
$\ell^i+t_i>m$, then we set the approval count of the $i$-th voter to be
$0$ or $m$, respectively.  We set $\Sigma(\vect) = \sum_{i=1}^n \sigma^i (t_i)$.

We say that a push action $\vect=(t_1, \dots, t_n)$ is {\em minimal}
for $I$ if $\Sigma(\vect)=\opt(I)$, $p$ is a winner in $\push(C, V, \vect)$,
and for every push action $\vecs\neq\vect$ such that $0 \le s_i\le t_i$ or $t_i \le s_i \le 0$
for all $i=1, \dots, n$ it holds that
$p$ is not a winner in $\push(C, V, \vecs)$.
Note that, as in the case of shift
bribery, an optimal push action is not necessarily minimal because it
may perform unnecessary zero-cost pushes. 

We also consider two natural special types of support bribery cost functions.
We say that a support bribery cost function $\sigma$ is {\em positive} if
$\sigma(k)=+\infty$ for each $k<0$, and we say that it is {\em
  negative} if $\sigma(k)=+\infty$ for each $k>0$.  The support
bribery problem with positive cost functions corresponds to the
setting where the campaign manager can only increase the voters'
approval counts, and can be viewed as a fine-grained version of
control by adding voters; similarly, the support bribery with negative
cost functions can be viewed as a refinement of control by deleting
voters (see the survey of Faliszewski, Hemaspaandra, and
Hemaspaandra~\cite{fal-hem-hem:j:cacm-survey} for a discussion of the complexity of election control and
for further references).

To conclude this section, we
observe that we have defined shift bribery and support bribery
as function problems. However, when talking about $\np$-completeness, 
we consider the decision variants of these problems, where 
we ask if there exists a successful bribery whose total cost does not exceed
a given value $b$---the {\em bribery budget}.

\section{Shift Bribery}
\label{sec:shift}

In this section, we present our results for {\sc shift bribery} 
under the Bucklin rule and the Fallback rule.
We start by describing our algorithm for the simplified version of the 
Bucklin rule; this algorithm can be modified to work for 
the simplified version of the Fallback rule. 

\begin{theorem}\label{thm:sim-B}
  Simplified Bucklin-{\sc shift bribery} is in $\p$. %
\end{theorem}
\begin{proof}
  Given an instance $I=(C, V, \Pi, p)$ of Simplified Bucklin-{\sc shift bribery}, 
  let $m=|C|$, $n=|V|$, and let $k$ be the Bucklin winning round for
  $(C, V)$. Let $W\subseteq C$ be the set of the
  simplified Bucklin winners in $(C, V)$. We can assume that $p\not\in W$.

  Let $\vect=(t_1, \dots, t_n)$ be a minimal optimal shift action for
  $I$.
  Let $\ell$ be the Bucklin winning round in $\shift(C, V, \vect)$.
  We claim that $\ell\in\{k, k+1\}$.  Indeed, any shift action moves every
  candidate in $W$ by at most one position downwards in each vote. Therefore, in
  $\shift(C, V, \vect)$ all candidates in $W$ win in round $k+1$, and
  hence $\ell\le k+1$.  Now, suppose that $\ell<k$. In $(C, V)$ the
  $\ell$-Approval score of every candidate is at most
  $\lfloor\frac{n}{2}\rfloor$, so the only candidate that can win in
  round $\ell$ in $\shift(C, V, \vect)$ is $p$, and for that she has
  to be moved into position $\ell$ in some voters'
  preferences.  However, moving $p$ into position $k$ in those voters'
  preferences suffices to make $p$ a winner in round $k$ (and
  thus an election winner), and we have assumed that $\vect$ is
  minimal.  This contradiction shows that $\ell\ge k$.
  Hence, to find an optimal shift bribery, it suffices to (a) compute
  a minimum-cost shift action that makes $p$ a winner in round $k$, (b)
  compute a minimum-cost shift action that makes $p$ a winner in round
  $k+1$ and ensures that no other candidate wins in round $k$, and (c)
  output the cheaper of the two.

  To win in round $k$, $p$ needs to obtain
  $\lfloor\frac{n}{2}\rfloor+1-s_k(p)$ additional $k$-Approval points.
  Thus, to find a minimum-cost shift bribery that makes $p$ win in round
  $k$, we consider all votes in which $p$ is not ranked in the top $k$
  positions, order them by the cost of moving $p$ into the $k$-th
  position (from lowest to highest), and pick the first
  $\lfloor\frac{n}{2}\rfloor+1-s_k(p)$ of these votes. Let $\vecs$
  denote the shift action that moves $p$ into position $k$ in each of
  those votes.

  Computing a shift action that ensures $p$'s victory in the $(k+1)$-st
  round is somewhat more difficult. In this case we need to ensure that
  (a) each candidate in $W$ is demoted from position $k$ to position $k+1$
  enough times that it does not win in round $k$, and (b)
  $p$'s $(k+1)$-Approval score is at least
  $\lfloor\frac{n}{2}\rfloor+1$. %
  Thus, we need to find an optimal balance between bribing
  several groups of voters.

  For each $c \in C \setminus\{p\}$, let $V_c$ denote the set of all
  voters that rank $c$ in the $k$-th position and rank $p$ below $c$;
  note that $c\neq c'$ implies $V_{c}\cap V_{c'}=\emptyset$.  
  Let us fix a candidate $c$ in $C
  \setminus \{p\}$. The only way to ensure that $c$ does not win in round $k$
  is to shift $p$ into position $k$ in at least
  $n(c)=\max\{0,s_k(c)-\lfloor\frac{n}{2}\rfloor\}$ votes in $V_c$. Note
  that $n(c)>0$ if and only if $c \in W$.  Thus, if for some $c\in W$
  we have $|V_c|<n(c)$, there is no way to ensure that 
  $c$ does not win in round $k$, so in this case we output
  $\vecs$ and stop.

  Otherwise, we proceed as follows. Let $A_c$ be
  the set of all voters in $V_c$ that rank $p$ in position $k+1$,
  and let $B_c = V_c \setminus A_c$. Note that for each vote in $A_c$,
  shifting $p$ into the $k$-th position does not change the
  $(k+1)$-Approval score of $p$, while doing the same for
  a vote in $B_c$ increases the $(k+1)$-Approval score of $p$ by
  one. For each $i=0, \dots, |B_c|$, let $b(c,i)$ be the
  minimum cost of a shift action that
  (a) shifts $p$ into position $k+1$ or above in $i$ votes from $B_c$, and
  (b) shifts $p$ into position $k$ in at least $n(c)$ votes
  from $A_c \cup B_c$.

  We can compute the numbers $b(c,i)$ for all $c\in C\setminus\{p\}$
  and all $i=0, \dots, |B_c|$
  using dynamic programming, as follows.
  Fix a candidate $c\in C\setminus\{p\}$. Reorder the voters 
  so that the voters in $B_c$ appear first, ordered 
  according to their cost of moving $p$ into the $(k+1)$-st position
  (from lowest to highest), followed by the voters 
  in $A_c$, ordered
  according to their cost of moving $p$ into the $k$-th position    
  (from lowest to highest). 
  After this step the $j$-th voter in $B_c$ is $v^j$.
  For each $i$ and $j$, $0 \leq i \leq j \leq |B_c|$, 
  and each $h=0, \dots, n(c)$, we define $b(c,i,j,h)$
  to be the cost of a minimum-cost shift action that only involves
  the voters in $A_c$ and the first $j$ voters in $B_c$
  and that
  (a) shifts $p$ into position $k+1$ or above in $i$ votes from $B_c$, and
  (b) shifts $p$ into position $k$ in at least $h$ votes from $A_c \cup B_c$.
  If there is no such shift action, we set $b(c,i,j,h)=+\infty$.
  
  We can compute $b(c,0,j,h)$ for all $j=0, \dots, |B_c|$ and all $h=0, \dots, n(c)$
  by bribing the first $h$ voters in $A_c$ to shift $p$ into position $k$.
  Similarly, we can compute $b(c,i,j,0)$ for all $0\le i\le j\le|B_c|$ 
  by bribing the first $i$ voters in $B_c$ to shift $p$ into position $k+1$.
  For all the remaining cases, we compute $b(c,i,j,h)$ using the
  following formula:
  \begin{equation}
    \label{eq:buck:dp}
    b(c,i,j,h) = 
	\min \left \{
    \begin{array}{l}
    b(c,i-1,j-1,h) + \pi^{j}(\rank(p,v^{j}) - (k+1)), \\
    b(c,i-1,j-1,h-1) + \pi^{j}(\rank(p,v^{j}) - k), \\	
    b(c,i,j-1,h).
    \end{array}
		\right.
  \end{equation}
  The first two lines of this formula correspond to the cases
  where the $j$-th voter in $B_c$ is bribed to shift candidate $p$ into positions $k+1$
  and $k$, respectively. The third line deals with the
  case where this voter is not bribed.  It is immediate
  that this method correctly computes the desired values. By
  definition, we have $b(c,i)=b(c,i,|B_c|,n(c))$.  
  For each candidate
  $c \in C \setminus\{p\}$ and each $i=0, \dots, |B_c|$, we define
  $\vecr(c,i)$ to be the shift action corresponding to the value
  $b(c,i)$; this shift action can be extracted from the dynamic programming computation 
  of $b(c,i)$ using standard techniques.

  Let $V_p=V\setminus \bigcup_{c\in C\setminus p} V_c$; the set $V_p$
  consists of all voters who rank $p$ in top $k$ positions.
  Note that in any minimal shift action, voters in $V_p$ are not bribed.
  Now, for every $j=1, \dots, m-1$ and every $i = 0, \dots, \lfloor\frac{n}{2}\rfloor+1-s_{k+1}(p)$,
  let $\beta(j, i)$ be the minimum cost of a shift action 
  that (i) only involves voters in $\cup_{\ell=1}^j V_{c_j}$, (ii) ensures that candidates
  $c_1, \dots, c_j$ do not win in round $k$ and (iii) ensures that $p$'s 
  $(k+1)$-Approval score is at least $s_{k+1}(p)+i$; we set $\beta(j, i)=+\infty$ 
  if no such shift action exists. The numbers $\beta(j, i)$, 
  $j=1, \dots, m-1$, $i = 0, \dots, \lfloor\frac{n}{2}\rfloor+1-s_{k+1}(p)$, 
  can be computed by dynamic programming as follows. 
  We have 
  $$
  \beta(1, i)=
  \begin{cases}
   b(c_1, i) &\text{ for }i=0, \dots, n(c_1),\\
  +\infty              &\text{ for }i = n(c_1)+1,\dots, \lfloor\frac{n}{2}\rfloor+1-s_{k+1}(p).
  \end{cases}
  $$
  Further, for every $j>1$ and all $i = 0, \dots, \lfloor\frac{n}{2}\rfloor+1-s_{k+1}(p)$ we have
  $$
  \beta(j, i) = \min\left\{\beta(j-1, i-\ell)+b(c_j, \ell)\mid \ell=0,\dots,\min(i, n(c_j))\right\}.
  $$    
  By construction, $\beta(m-1,\lfloor\frac{n}{2}\rfloor+1-s_{k+1}(p))$
  is the minimum cost of a shift action ensuring that 
  $p$ wins in round $k+1$ and no other candidate wins in round $k$.
  A shift action $\vecr$ that has this cost can be extracted from the dynamic programming
  computation of $\beta(m-1,\lfloor\frac{n}{2}\rfloor+1-s_{k+1}(p))$; it is the sum 
  of shift actions $\{\vecr(c_j, i_j)\}_{c_j\in C\setminus p}$ for appropriate values
  of $i_1,\dots, i_{m-1}$. 

  We output the cheaper of $\vecs$ and $\vecr$.  This
  algorithm clearly runs in polynomial time, and our argument shows
  that it produces an optimal shift action for $I$.
\end{proof}

\begin{theorem}\label{thm:sim-fallback}
  Simplified Fallback-{\sc shift bribery} is in $\p$. %
\end{theorem}
\begin{proof}
  Given an instance $I=(C, V, \Pi, p)$ of Simplified Fallback-{\sc shift bribery},
  let $m=|C|$ and  $n=|V|$. Further, let $\pref^i$ be the
  preference order of voter $v^i$, and let $\ell^i$ be her approval count. 
  Set $L=\max_{i=1, \dots, n} \ell^i$.

  Our algorithm proceeds in two stages. First, it computes a shift action
  $\vecs$ of minimum cost that (after deleting the non-approved candidates
  from each voter's preference list) ensures that there is no winner
  according to the simplified Bucklin rule and that no candidate 
  has more approvals in total than the preferred candidate $p$. 
  Second, for each $\ell=1, \dots,  L$ 
  it computes a shift action $\vecr^{\ell}$ of minimum cost ensuring
  that $p$ wins under simplified Bucklin in round $\ell$
  (that is, the Bucklin winning round is $\ell$, and $p$ is an $\ell$-Approval winner).
  Finally, we output the cheapest among $\vecs$ and $\vecr^{\ell}$,
  $\ell=1, \dots, L$.

  In what follows, we say that $p$ can \emph{exclude $c$ from
  round $\ell$} in a vote $v^i$ if $c \pref^i p$ and 
  $\rank(c,v)=\min (\ell^i, \ell)$.
  We say that $p$ can \emph{exclude $c$ from approval}
  if $c \pref^i p$ and $\rank(c,v)=\ell^i$.

  To obtain $\vecs$, it suffices to find, for each $t=1,\dots, \lfloor n/2 \rfloor$,
  a minimum-cost shift action $\vecs^t$ %
  ensuring that $p$ receives $t$ approvals, and everybody else receives
  at most $t$ approvals. This can be done greedily, as follows.
  For each candidate $c \in C \setminus \{p\}$, 
  let $V_c$ be the set of votes where $p$ can exclude $c$ from
  approval. For each $c \in C \setminus \{p\}$ with $s(c)>t$,
  we order the votes in $V_c$ according to their cost $\pi^i(\rank(p,v^i)-\rank(c,v^i))$,
  and bribe the cheapest $s(c)-t$ voters in $V_c$ (if $|V_c|<s(c)-t$ for some such $c$, 
  then $\vecs^t$ remains undefined). 
  If after this stage $p$ has $t'<t$ approvals,
  we bribe the cheapest $t-t'$ voters in $\bigcup_{C \setminus \{p\}}V_c$ that have not been bribed yet 
  (again, if the number of such voters is less than $t-t'$, then $\vecs^t$ remains undefined).
  We then set $\vecs$ to be the cheapest among $\vecs^1, \dots, \vecs^{\lfloor n/2 \rfloor}$
  (where the cost of an undefined shift action is taken to be $+\infty$).  
  
  The computation of $\vecr^1$ is easy: we just have to make sure
  that $p$ receives at least $\lfloor n/2 \rfloor +1$ points in
  the first round, so we can sort the voters according to the cost 
  of shifting $p$ into the top position (from lowest to highest),
  and bribe the first $\lfloor n/2 \rfloor +1-s_1(p)$ voters.  
  To compute $\vecr^{\ell}$ for $\ell=2, \dots, L$,
  we employ the algorithm used for computing the shift action $\vecr$ 
  in the proof of Theorem~\ref{thm:sim-B}, 
  with a few modifications.
  Specifically, for each $c \in C\setminus \{p\}$ we let $V_c$
  contain the votes in which $p$ can exclude $c$ from round $\ell-1$.
  We partition $V_c$ into $A_c$ and $B_c$ by setting 
  $A_c = \{v^i\in V_c\mid \rank(c,v^i)=\ell-1, \rank(p,v^i)= \ell \leq \ell^i\}$ 
  and $B_c = V_c \setminus A_c$. For $i=1, \dots, n$, $j=1, \dots, m$
  we say that position $j$ in vote $v^i$ is {\em $\ell$-good} if
  $j\le \min(\ell^i, \ell)$.
  We then proceed as in the proof of Theorem~\ref{thm:sim-B}:
  to determine $\vecr^{\ell}$, we compute for each 
  $i=0, \dots, |B_c|$ a minimum-cost shift action that shifts $p$
  into an $\ell$-good position in $i$ votes from $B_c$ and %
  excludes $c$ from round $\ell-1$ in at least $s_{\ell-1}(c)-
  \lfloor n/2 \rfloor$ votes from $A_c \cup B_c$, and then use dynamic programming
  to decide how many voters from each set $V_c$, $c\in C\setminus\{p\}$, 
  we want to bribe.
\end{proof}

Shift bribery also admits a polynomial-time algorithm for the regular version of the Bucklin rule;
however, the proof becomes more involved.

\begin{theorem}\label{thm:bucklin}
  Bucklin-{\sc shift bribery} is in $\p$. %
\end{theorem}
\begin{proof}
  Let $m=|C|$, $n=|V|$, and let $k$ be the Bucklin winning round for
  $(C, V)$. Let $W$ denote the set of Bucklin winners in $(C, V)$.

  Let $\vect=(t_1, \dots, t_n)$ be a minimal optimal shift action for
  $I$, and let $\ell$ be the Bucklin winning round in $\shift(C, V,
  \vect)$.  We have $\ell\in\{k-1, k, k+1\}$.  Indeed, the argument in
  the proof of Theorem~\ref{thm:sim-B} shows that $\ell\le k+1$. Now,
  suppose that $\ell<k-1$.  This means that in $\shift(C, V, \vect)$
  our preferred candidate $p$ wins in round $\ell$.  Consider the set
  of all voters that were requested to shift $p$ into position $\ell$
  or higher under $\vect$; this set must be non-empty since prior to
  the bribery $p$ did not win in round $\ell$.  If we now demote $p$
  into position $k-1$ in those votes, she would still win in round
  $k-1$. Moreover, since no other candidate wins in round $k-1$ in the
  original election, $p$ is now the unique winner in round $k-1$ and
  hence the Bucklin winner, a contradiction with $\vect$ being a
  minimal optimal shift action.

  We will now find (a) a minimum-cost shift action that makes $p$ a
  winner in round $k-1$, (b) a minimum-cost shift action that makes $p$
  a winner in round $k$ and ensures that no candidate has a higher
  $k$-Approval score than $p$, and (c) a minimum-cost shift action that
  makes $p$ a winner in round $k+1$ and ensures that no other candidate
  wins in round $k$ or has a higher $(k+1)$-Approval score than $p$.
  We will then output the cheapest of these three shift actions.

  The first step is straightforward: we order all voters that do not
  rank $p$ in the first $k-1$ positions by the cost of shifting $p$
  into position $k-1$ in their votes (from the lowest to the highest),
  and bribe the first $\lfloor\frac{n}{2}\rfloor+1-s_{k-1}(p)$ of them
  to move $p$ into position $k-1$ in their votes. Denote this
  shift action by $\vecs$.

  The second step is somewhat more difficult. Namely, for each
  $i=\lfloor\frac{n}{2}\rfloor+1, \dots, n$, let $\vecr^i$ be 
  a minimum-cost shift action that ensures that $p$'s $k$-Approval score is
  at least $i$, and the $k$-Approval score of any other candidate is
  at most $i$.  We can compute $\vecr^i$ as follows.

  Recall that $W$ is the set of Bucklin winners in election $(C,V)$.
  For each candidate $c\in W$, let $V_c$ denote the set of all voters
  that rank $c$ in the $k$-th position and rank $p$ below $c$.  To
  ensure that $c$'s $k$-Approval score is at most $i$, we need to
  shift $p$ into position $k$ in at least $s_k(c)-i$ such votes. Thus
  if for some $c\in W$ we have $|V_c|<s_k(c)-i$ then we record that
  for this value of $i$ the shift action $\vecr^i$ is undefined and we
  set its cost to $+\infty$.  Otherwise, we order the votes in each
  $V_c$ by the cost of moving $p$ into the $k$-th position in this
  vote (from lowest to highest), and bribe the first $s_k(c)- i$
  voters in each set to move $p$ in the $k$-th position in their
  votes; denote the corresponding shift action by $\vecr^{i, 1}$.  In
  the resulting election $\shift(C, V, \vecr^{i, 1})$, no candidate
  other than $p$ gets more than $i$ $k$-Approval points.  Let $s'_{k}$
  be $p$'s $k$-Approval score in $\shift(C, V, \vecr^{i, 1})$.  If
  $s'_{k} \ge i$,
  we set
  $\vecr^i=\vecr^{i, 1}$.  Otherwise, we order the voters that rank
  $p$ in position $k+1$ or lower in $\shift(C, V, \vecr^{i, 1})$ by
  the cost of moving $p$ into position $k$ in their preferences (from
  lowest to highest), and bribe the first $i-s'_{k}$ of them to move
  $p$ into position $k$ in their votes. Denote this bribery by
  $\vecr^{i, 2}$, and set $\vecr^i=\vecr^{i, 1}+\vecr^{i, 2}$.
  Finally, let $\vecr$ be the cheapest shift action among
  $\vecr^{\lfloor\frac{n}{2}\rfloor+1}, \dots, \vecr^n$.

  Now, finding a minimum-cost shift action that makes $p$ win in round
  $k+1$ and ensures that no other candidate wins in an earlier round
  or has more $(k+1)$-Approval points than $p$ is yet more
  difficult. Indeed, we must balance the need to demote the candidates
  that may win in round $k$ against the need to demote the candidates
  that may beat $p$ in round $k+1$.  Note also that we may need to
  shift $p$ into position $k$ in some votes in order to lower the
  $k$-Approval score of its competitors.

  We deal with these issues by reducing our problem to that of finding
  a minimum-cost circulation. Recall that an instance of a
  minimum-cost circulation problem is given by a directed graph
  $\calG=(\calV, \calE)$, and, for each $(v, w)\in\calE$, a lower bound
  $l(v, w)$ on the flow from $v$ to $w$, an upper bound $u(v, w)$ on
  the flow from $v$ to $w$ and the cost $c(v, w)$ of a unit of flow
  from $v$ to $w$.  A solution is a feasible flow, i.e., a vector
  $(f(v, w))_{(v, w)\in\calE}$ that satisfies (a) $l(v, w)\le f(v, w)
  \le u(v, w)$ for all $(v, w)\in\calE$ and (b) $\sum_{(z,
    v)\in\calE}f(z, v) = \sum_{(v, w)\in\calE}f(v, w)$ for any
  $v\in\calV$.  The cost of a feasible solution is given by
  $$
  \sum_{(v, w)\in\calE}c(v, w)f(v, w).
  $$
  An optimal solution is one that minimizes the cost among all
  feasible solutions.  It is well-known that when all costs and
  capacities are integers, an optimal circulation is integer and can be
  found in polynomial time 
  (see, e.g., \cite{eds-kar:j:circulation,tar:j:combinatorica}).

  Given an instance of our problem, we construct a family of instances
  of minimum-cost circulation, one for each 
  $i = \max\{0,\lfloor\frac{n}{2}\rfloor+1-s_{k+1}(p)\}, \dots, n$  
  as follows (see Figure~\ref{fig:flow}).  For each $i$,
  our graph $\calG^i$ models the situation where we bribe exactly $i$
  voters ranking $p$ in position $k+2$ or lower.
  We let $\calG^i$ consist of six ``layers''.  The first layer
  consists of a single vertex $S$, and the second layer consists of a
  single vertex $S'$.  In the third layer, we have a vertex $U_h$
  for each candidate $c_h\in C\setminus\{p\}$.  In the fourth
  layer, we have a vertex $W_j$ for each $j=1, \dots, n$.  In the
  fifth layer, we have a vertex $Z_h$ for each candidate $c_h\in
  C\setminus\{p\}$.  The sixth layer consists of a vertex $T$.

\begin{figure}[ht]
  \begin{center}	
    \resizebox{10cm}{!}{
    \includegraphics{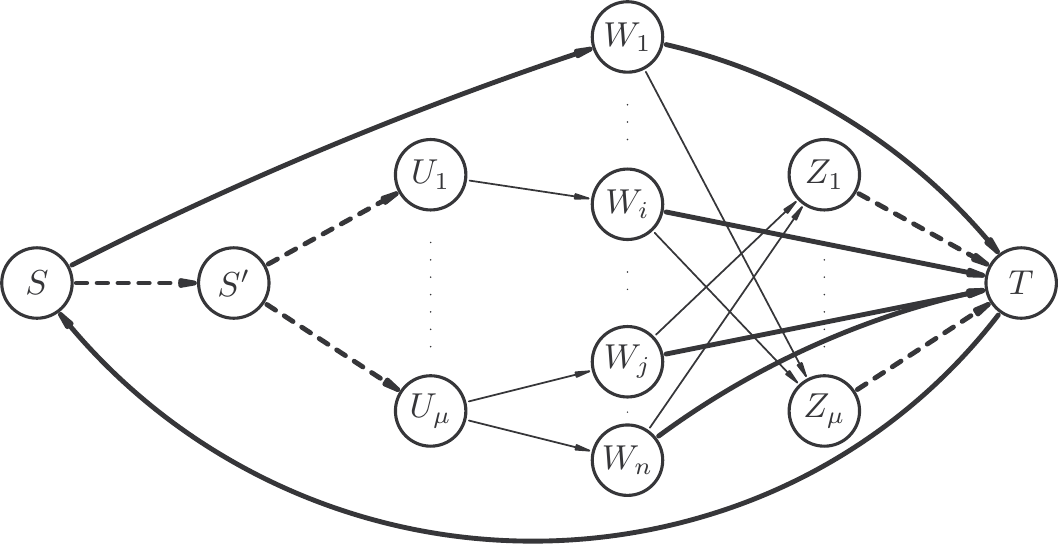}
}
    \caption{\label{fig:flow}Graph $\calG^i$.  For readability, we set
      $\mu=m-1$. The bold arcs are unconstrained, the dashed arcs
      have a lower bound on the size of the flow, and the regular
      arcs can carry at most one unit of flow. The graph corresponds
      to an instance 
      where the voters' preferences are such that 
      $\rank(p,v^1)=\rank(c_1,v^i)=\rank(c_{m-1},v^j)=\rank(c_{m-1},v^n)= k+1$, 
      and moreover,
      $\rank(c_{m-1},v^1)=\rank(c_{m-1},v^i)=\rank(c_1,v^j)=\rank(c_1,v^n)=k$.
	  }
  \end{center}
\end{figure}

The costs and capacities of the arcs depend on the value of $i$. We
will now describe them layer-by-layer.  In our description, we will
say that an arc $(v, w)$ is {\em unconstrained} if it satisfies $l(v,
w)=0$, $u(v, w)=+\infty$, and $c(v, w)=0$.

There is an arc $(S, S')$ with $l(S, S')=u(S, S')=i$ and $c(S,
S')=0$.  Also, for each $c_h\in C\setminus\{p\}$, there is an arc
from $S'$ to $U_h$ with $l(S',U_h)=\max \{ 0,s_{k+1}(c_h)-s_{k+1}(p)-i \}$,
$u(S', U_h)=+\infty$, $c(S',U_h)=0$.

Each vertex $W_j$, $j=1, \dots, n$, in the fourth layer has one
incoming arc if $v^j$ ranks $p$ in position $k+1$ or lower, and no
incoming arcs otherwise. Specifically, if $v^j$ ranks some $c_h\in
C\setminus\{p\}$ in position $k+1$, and ranks $p$ in position $k+2$ or
lower, then the graph contains an arc $(U_h, W_j)$ with $l(U_h,
W_j)=0$, $u(U_h, W_j)=1$.  We set $c(U_h, W_j)$ to be equal to
the cost of shifting $p$ into position $k+1$ in $v^j$. If $v^j$ ranks
$p$ in position $k+1$, there is an unconstrained arc from $S$ to $W_j$.

There are two types of arcs leaving the vertices in the fourth layer.
First, for each $j=1, \dots, n$, there is an unconstrained arc $(W_j,
T)$.  Second, there is an arc $(W_j, Z_h)$ if and only if $v^j$
ranks candidate $p$ in position $k+1$ or lower, and $c_h$ is
ranked in position $k$. We set $l(W_j, Z_h)=0$, $u(W_j, Z_h)=1$,
and let $c(W_j, Z_h)$ to be equal to the cost of shifting $p$ from
position $k+1$ into position $k$ in the preferences of the $j$-th
voter, i.e., $c(W_j, Z_h) = \pi^j(rank(p, v^j)-k)-\pi^j(rank(p, v^j)-(k+1))$. %

For each vertex $Z_h$ in the fifth layer, there is an arc from
this vertex to $T$ that satisfies $l(Z_h,
T)=\max \{0, s_k(c_h)-\lfloor\frac{n}{2}\rfloor \}$, $u(Z_h, T)=+\infty$,
$c(Z_h, T)=0$.  Finally, there is an unconstrained arc $(T, S)$.

In this construction, the flow value on the arc $(T,S)$ corresponds to the number of
voters bribed.  Specifically, the flow from $S$ to $S'$ reflects the
number of voters that rank $p$ in position $k+2$ or lower that were
bribed to shift $p$ into position $k+1$ or higher, while the flow 
going directly from
$S$ to the vertices in the fourth layer reflects the number of voters
that ranked $p$ in position $k+1$ and were bribed to shift $p$ into
position $k$.  Now, sending $x$ units of flow from $S'$ to $U_h$
corresponds to shifting $p$ into the $(k+1)$-st position in $x$ votes
that rank $c_h$ in the $(k+1)$-st position. The lower bound on the
flow ensures that we bribe at least $s_{k+1}(c_h)-s_{k+1}(p)-i$
such voters to move $p$ in position $k+1$ or higher, and hence after
the bribery $c_h$'s $(k+1)$-Approval score is at most
$s_{k+1}(p)+i$. Note that if we bribe exactly $i$ voters that do not
rank $p$ in the top $k+1$ positions to shift $p$ into position $k+1$
or higher, $p$'s $(k+1)$-Approval score becomes exactly $s_{k+1}(p)+i$, so
these arcs ensure that no candidate has a higher $(k+1)$-Approval
score than $p$. 

The arcs between the third and the fourth layer ensure that such a
bribery can actually be implemented (i.e., for each $c_h\in C \setminus \{p\}$,
there are sufficiently many voters that rank $c_h$ in the
$(k+1)$-st position and rank $p$ below $c_h$); their cost reflects
the cost of moving $p$ into position $k+1$ in the preferences of the
bribed voters.

The flow leaving the fourth layer and going directly into $T$
corresponds to the voters bribed to shift $p$ into position $k+1$
only, while the flow between the fourth and the fifth layer
corresponds to the voters that were bribed to move $p$ into position
$k$.  The cost of the arcs between the fourth and the fifth layer
corresponds to the cost of shifting from position $k+1$ to position
$k$; note that the cost of shifting $p$ into position $k+1$ in that
vote has already been accounted for.

The arcs between the fifth and the sixth layer ensure that no
candidate wins in round $k$: the flow of size $t$ from $Z_h$ to $T$
signifies that $t$ voters that rank $c_h$ in position $k$ have been
bribed to move $p$ into position $k$ (and therefore demote $c_h$
into position $k+1$). Thus, by satisfying all flow constraints, we
ensure that the $k$-Approval score of each candidate in
$C\setminus\{p\}$ is at most $\lfloor\frac{n}{2}\rfloor$.

The argument above shows that each valid circulation in $\calG^i$, $i
= \max\{0, \lfloor\frac{n}{2}\rfloor+1-s_{k+1}(p)\}, \dots, n$,
corresponds to a successful shift bribery in which exactly $i$ voters
that rank $p$ in position $k+2$ or lower were bribed to shift $p$ into
position $k+1$ or higher, and the cost of the circulation is equal to
the cost of this bribery. Moreover, the corresponding shift action can
be easily computed given the circulation.  Clearly, the converse is
also true: if there is a successful shift action of cost $X$ that
bribes exactly $i$ voters that rank $p$ in position $k+2$ or lower,
there is also a valid circulation of size $i$ that has the same cost.

For each $i = \max\{0, \lfloor\frac{n}{2}\rfloor+1-s_{k+1}(p)\},
\dots, n$, let $\vecq^i$ denote the shift action that corresponds to
an optimal circulation in $\calG^i$ if one exists; if there is no
valid circulation in $\calG^i$, we leave $\vecq^i$ undefined and set
its cost to $+\infty$.  Now, let $\vecq$ be a minimum-cost shift action
among $\vecq^i$; if all $\vecq^i$ are undefined, $\vecq$ remains
undefined as well, and its cost is set to $+\infty$.  
By construction, $\vecq$ is a minimum-cost shift action that makes $p$ a winner
in the $(k+1)$-st round and ensures that no candidate wins in the
$k$-th round or has more $(k+1)$-Approval points than~$p$.

Finally, consider the shift actions $\vecs$, $\vecr$, and $\vecq$, and
output one with the minimum cost.  This algorithm runs in polynomial time,
and outputs a minimum-cost shift action that makes $p$ a winner.
\end{proof}

A similar approach works for the Fallback rule.

\begin{theorem}\label{thm:fallback}
  Fallback-{\sc shift bribery} is in $\p$. %
\end{theorem}
\begin{proof}
  The proof is very similar to that of Theorem~\ref{thm:bucklin}. We
  use essentially the same algorithm, but make the following
  modification.  As in the proof of Theorem~\ref{thm:sim-fallback}, we
  compute the cost of a shift bribery that ensures that no candidate
  is approved by more than half of the voters and our candidate $p$ wins
  by approval count (if such a shift bribery exists).  Then, we apply
  the algorithm from Theorem~\ref{thm:bucklin}, taking into account
  that under Fallback voting it is sometimes possible to ensure that
  some candidate $c$ becomes disapproved after $p$ is shifted
  into her position.  Modifying the algorithm from
  Theorem~\ref{thm:bucklin} to take advantage of this possibility is
  straightforward, and we omit a detailed argument.
\end{proof}

\section{Support Bribery}
\label{sec:support} 

In the technical report version of their work, Elkind et al.~\cite{elk-fal-sli:c:swap-bribery}
prove an NP-completeness result for mixed bribery under SP-AV. 
Their proof does not rely on shifting the preferred
candidate in the voters' preferences, and therefore applies to support
bribery as well, showing that the decision version of
SP-AV-\textsc{support bribery} is $\np$-complete. In this section we
extend this result to Fallback voting, and explore the parameterized
complexity of support bribery under both the simplified and the
classic variant of this rule.

Each instance $I$ of support bribery can be associated with the
following parameters. First, let $\alpha(I)$ denote the maximum number
of bribed voters over all minimal briberies that solve $I$ optimally.
Second, let $\beta(I)$ and $\beta'(I)$ denote, respectively, the
maximum and the minimum of $\sum_{i=1}^n |t_i|$ over all minimal briberies
$(t_1, \dots, t_n)$ that solve $I$ optimally; these parameters describe 
the total change in the approval counts.  Observe that
$\beta(I) \geq \beta'(I)$ and $\beta(I) \geq \alpha(I)$ for every
instance $I$.

We will now demonstrate that support bribery under Fallback voting is
computationally hard, even in very special cases.  These results,
while somewhat disappointing from the campaign management perspective,
are hardly surprising. Indeed, we have argued that support bribery can
be viewed as a fine-grained version of control by adding/deleting
voters, and both of these control problems are $\np$-hard for Fallback
voting~\cite{erd-fel-rot-sch:j:fallback-buclin-control-theory}. In fact, since Fallback voting defaults to
Approval voting if no candidate is approved by a majority of voters,
by introducing appropriate dummy candidates and voters we can
easily reduce the problem of control by adding voters under Approval
to the problem of support bribery under Fallback voting.

Our next result shows that \textsc{support bribery} is $\np$-hard
for both simplified Fallback voting and regular Fallback voting, 
even under very strong restrictions on the cost function;
moreover, these problems remain
intractable even for instances with a small value of $\alpha$.  Thus,
bribing even a few voters can be a hard task.

\begin{theorem}
  \label{thm:fallback-support-A}
  Both Fallback-\textsc{support bribery} and
  simplified Fallback-\textsc{support bribery} are $\np$-complete,
  and also {\em W[2]}-hard with respect to parameter $\alpha$ 
  (the maximum number of voters to be bribed),
  even in the special case where each cost is either
  $+\infty$ or $0$, and either all cost functions are positive
  or all cost functions are negative.
\end{theorem}
\begin{proof}
  For both types of cost functions (all-positive and all-negative), we give
  a polynomial-time computable parameterized reduction from the
  W[2]-hard \textsc{dominating set} problem;
  the reductions are inspired by those given by Erd\'elyi et al.~\cite{erd-fel-rot-sch:j:fallback-buclin-control-theory} in their proof of W[2]-hardness
  of control by adding/deleting voters in Fallback voting.

  We start by considering negative cost functions.  Let
  $G=(\calV,\calE)$ and $k$ be the given input for \textsc{dominating
    set}.  We assume $\calV=\{v_1, \dots, v_n\}$ and we write $N[v_i]$
  to denote the closed neighborhood of vertex $v_i$ in $G$. 
  We then construct an election $E$ as follows.
  Our set of candidates is $\calV \cup \{a,b,p\} \cup D$, where $p$ is the preferred
  candidate and $D$ is a set of dummy candidates. We do not specify
  the set of the dummy candidates explicitly; rather, we require that
  it is large enough that each of the dummies is approved by at most
  one voter. It will become clear in the course of the proof that choosing
  a polynomial-size set $D$ with this property is possible. There are
  $6n$ voters, and each of them approves $n+3$ candidates. 
  Specifically, for each $v_i \in \calV$, we construct two voters, $x_i$ and $\bar{x}_i$, and we
  construct additional $4n$ votes in order to adjust the scores of
  the candidates for our purposes.  The preferences of the voters are
  shown below. We use dots to denote dummies, and we use sets in the
  lists when their elements can be ordered arbitrarily.  The sign $|$
  indicates the approval count; non-approved candidates are not
  listed.

  \begin{tabbing}
  123\=123456789012345\=12345678901234567890\=\kill
  \> $\textrm{voter $x_i$:}$ \> $a \pref N[v_i] \pref \dots$ \> $\pref p$ \= $\pref b \;|$ \\
  \> $\textrm{voter $\bar{x}_i$:}$ \> $a \pref \calV \setminus N[v_i] \pref \dots$ \> $\pref p$ \> $\pref b \;|$ \\
  \> $\textrm{$2n+1$ voters:}$ \> $\calV \pref \dots$ \> \> $\pref b \;|$ \\
  \> $\textrm{$n+k$ voters:}$ \> $a \pref \dots$ \> $\pref p$ \> $\pref b \;|$ \\
  \> $\textrm{$1$ voter:}$ \> $\dots $ \> $\pref p$ \> $\pref b \;|$ \\
  \> $\textrm{$n-k-2$ voters:}$ \> $\dots $\> \> $\pref b \;|$ 
  \end{tabbing}

  The cost of decreasing the approval count arbitrarily is $0$ for each
  of the votes in $X=\{x_1, \dots, x_n\}$, and is $+\infty$ for all
  other votes; our budget is $0$.  Note that in election $E$ we have
  $s_{1}(a)=3n+k$, $s_{n+1}(v_i)=3n+1$ for each $i$,
  $s_{n+2}(p)=3n+k+1$, and $s_{n+3}(b)=6n$.

  Let $\vect$ be some minimal successful bribery for $I$.  Note that
  applying $\vect$ must decrease $a$'s first-round score by at least
  $k$; otherwise $a$ would be the unique winner of election
  $\push(E,\vect)$ under (simplified) Fallback. Thus, applying $\vect$
  sets the approval count to $0$ in at least $k$ votes in
  $X$. This decreases $p$'s score in round $n+2$ by at least $k$ points as
  well, so $p$ can have at most $3n+1$ points in round $n+2$ in
  $\push(E,\vect)$.
  On the other hand, in round $n+3$ candidate $b$ will have more
  approvals than $p$ under any bribery of cost $0$, 
  so $p$ can become a winner only if it wins in
  round $n+2$. Therefore, $p$'s score in round $n+2$ must remain at
  least $3n+1$ in $\push(E,\vect)$, which means that $\vect$ must set
  the approval count to $0$ in {\em exactly} $k$ votes from $X$.  
  Let $S$ be the set of these votes, and let
  $\{s_1, \dots, s_k\}$ be the corresponding vertices of $G$.
  As only voters in $X$ can be bribed within the budget and $\vect$
  is a minimal successful bribery, it follows that voters not in $S$
  are not bribed. Consequently, we have $\alpha(I)=k$.  

  Now, observe that no matter how $S$ is chosen, $a$ does not win in the first 
  $n+1$ rounds in $\push(E,\vect)$, and $p$ gets a strict majority of votes in round $n+2$.
  Therefore, $p$ wins in $\push(E,\vect)$ if and only
  if none of the candidates in $\calV$ gets $3n+1$ points in round $n+1$.
  This happens if and only if each vertex loses at least one point as a result of the 
  push action $\vect$, meaning that the sets $N[s_1], \dots, N[s_k]$ cover $\calV$.
  Since this occurs if and only if the vertices
  $s_1,\dots,s_k$ form a dominating set, we have proved the
  correctness of the reduction.

  We will now consider positive cost functions.
  Again, the reduction is from \textsc{dominating set}. Let $G=(\calV,\calE)$
  and $k$ be the input instance; we use the notation defined above. 
  Assume without loss of generality that $k\ge 2$. %
  We construct an election $E$ with candidate set $\calV \cup \{a,b,p\}$, where
  $p$ is the preferred candidate. The set of voters is of size
  $2n+2$, including a voter $x_i$ for each $v_i \in \calV$.
  Preferences and approval counts are shown below; we omit the non-approved
  candidates in the last $n+2$ votes.
  
  \begin{tabbing}
  123\=\kill
  \> $\textrm{voter $x_i$ with $0$ approvals:}$ \= $|\; \calV \setminus N[v_i] \pref b \pref p \pref a \pref N[v_i]$ \\
  \> $\textrm{$k$ voters with $1$ approval:}$ \> $a \;|$ \\
  \> $\textrm{$1$ voter with $n$ approvals:}$ \> $\calV \;|$ \\
  \> $\textrm{$n+1-k$ voters with $n+3$ approvals:} \; \; a \pref b \pref p \pref \calV \;|$
  \end{tabbing}

  The cost of increasing the approval count arbitrarily is $0$ in any
  of the votes in $X=\{x_1, \dots, x_n\}$, and is $+\infty$ in all
  other votes; our budget is $0$.  Observe that in $E$ we have $s_{1}(a)=n+1$,
  $s_{2}(b)=n+1-k$, $s_{3}(p)=n+1-k$, and $s_{n+3}(v_i)=n+2-k$ for
  each $i$. Thus, no candidate has strict majority (that
  is, $n+2$ points) in any round, so the winner is candidate $a$,
  who has the largest number of approvals.

  Let $\vect$ be some minimal successful push action for $I$.  Observe
  that since we can only increase the approval counts, $s_1(a)$ is
  $n+1$ in $\push(E,\vect)$ as well.  Hence, applying $\vect$ must
  increase $p$'s score by at least $k$ in order for $p$ to beat $a$ in some round.
  To achieve this, $\vect$ must increase the approval count in at
  least $k$ votes in $X$. On the other hand, suppose that we bribe more than $k$
  voters in $X$ to approve $p$. Then $b$ gets at least $n+2$ points in some 
  round of the resulting election; let $j$ be the first such round.
  As all voters prefer $b$ to $p$, it cannot be the case that $p$ 
  gets $n+2$ points in round $j$, so this bribery is not successful.
  Hence, applying $\vect$ must increase $p$'s approval count by {\em exactly} $k$, via bribing
  a subset of voters $S \subseteq X$ of size $k$. Note that this implies $\alpha(I)=k$ as
  well.  Let $S= \{s_1, \dots, s_k\}$.

  Now, looking at the scores of the candidates in $\calV$, one can see that 
  a support bribery that bribes voters in $S$ to increase $p$'s score is
  successful if and only if each $v_i \in \calV$ receives at most $k-1$ additional
  points from the voters in $S$. This holds if and only if each vertex is
  missing from at least one of the sets $\calV \setminus N[s_1], \dots, \calV \setminus N[s_k]$.
  This is equivalent to $s_1,\dots,s_k$ being a dominating set. Thus, the proof is complete.
\end{proof}

Since the hardness result of Theorem~\ref{thm:fallback-support-A} for
Fallback-\textsc{support bribery} holds even if all bribery costs are
either $0$ or $+\infty$, it follows that this problem does not admit
an approximation algorithm with a bounded approximation ratio.

Theorem~\ref{thm:fallback-support-A} shows that
Fallback-\textsc{support bribery} is W[2]-hard with respect to the
parameter $\alpha$. Given that we have $\beta(I) \geq \alpha(I)$ for
each instance $I$, it is natural to ask whether
Fallback-\textsc{support bribery} remains hard if even $\beta$ is
small, i.e., every minimal optimal bribery only makes small changes to the
approval counts.  In Section~\ref{sec:SP}, we will see that even a
very restricted version of this problem remains hard.  More precisely,
in Theorem~\ref{thm:fallback-support-SP} we will prove that
\textsc{support bribery} for each of SP-AV, simplified Fallback voting, and
Fallback voting remains $\np$-hard and also W[1]-hard with respect to parameter $\beta$,
even for single-peaked electorates and \emph{unit costs},
i.e., when $\sigma^i(k)=|k|$ for each $k$ and each $i=1, \dots, n$.

However, the hardness proof in Theorem~\ref{thm:fallback-support-SP} (see Section~\ref{sec:SP}) heavily
relies on the fact that unit cost functions allow us to
increase approval counts in some of the votes while decreasing
them in some other votes. In contrast, we will now prove that if all cost functions are positive
or all cost functions are negative, (simplified) Fallback-\textsc{support bribery}
is fixed-parameter tractable with respect to $\beta'$ (and hence also
with respect to $\beta$).
 
\begin{theorem}\label{thm:fallback-support-C}
\textsc{support bribery} for SP-AV, simplified Fallback voting, and Fallback voting 
is {\em FPT} with respect to parameter $\beta'$ (the minimum total change 
in approval counts over all optimal briberies), 
as long as either all bribery cost functions are positive
or all bribery cost functions are negative.  
\end{theorem}
\begin{proof}
  Suppose we are given an instance $I=(C,V,\Sigma,p)$ of
  \textsc{support bribery} with $|V|=n$. It will be convenient to assume that 
  we also given $\beta'=\beta'(I)$;
  if this is not the case, we can try each possible
    value of $\beta'$ in an increasing fashion.
  We will present an algorithm that works for both the simplified and
  the classic variant of Fallback voting; the algorithm can easily
  be modified (indeed, simplified) for SP-AV.

  Under the Fallback rule, a candidate can win by either 
  (a) having the highest number of approvals in the Bucklin winning round
  or (b) having the highest
  number of approvals when there is no Bucklin winning
  round; for simplified Fallback rule, condition (a) changes to
  (a$'$) obtaining a majority of approvals in the Bucklin winning round.
  To take into account briberies that ensure $p$'s victory
  via case (b), we view this case as an ``extra round'' in
  which the candidates with the highest number of approvals win. This
  way we can treat all cases in a uniform manner (it will be clear how
  to handle minor differences hidden by this notation).
  
  A bribery $\vect=(t_1,\dots,t_n)$ with
  $\sum_{i=1}^n |t_i| \leq \beta'$ has the following two properties, which will be used by our algorithms:
  \begin{enumerate}
  \item[(1)] $\vect$ can change the approval scores of at most $\beta'$
    candidates, and 
  \item[(2)] $\vect$ can change the approval score
    of each candidate by at most $\beta'$.  
  \end{enumerate}
  If $\vect$ makes $p$ win in round $\ell$ of the election
  $\push(C,V,\vect)$ while changing its 
  approval score by $\delta(p)$, then we say that $\vect$ is an
  \emph{$(\ell,\delta(p))$-bribery}; here $\ell$ may also refer to the
  extra round. As argued above, we can restrict
  ourselves to $(\ell,\delta(p))$-briberies with $0 \leq \delta(p)
  \leq \beta'$.\medskip
   
  \noindent
  \emph{Negative cost functions.}\quad 
  Suppose first that each cost function is negative; in this case any bribery can
  only decrease a candidate's score.
    
  We use a bounded search tree approach. By ``guessing'' an answer to
  a question, we always mean branching in the search tree according to
  all the possible ways of answering this question.  
  Our algorithm will branch at most $f(\beta')=\beta'+2$ times,
  and in each case it will branch into at most $g(\beta')=3^{\beta'+1}$
  directions. Moreover, 
  our algorithm will make at most a linear number of steps until
  reaching a leaf of the search tree. This ensures that 
  its running time can be bounded by $O(g(k)^{f(k)} |I|)=3^{O(\beta'^2)} |I|$,
  and hence our problem is fixed-parameter tractable with respect to
  $\beta'$.

  We first make some observations regarding our input instance.  If
  $p$ is approved by at most $\lfloor n/2 \rfloor$
  voters in $(C,V)$, then its only chance to win in
  $\push(C,V,\vect)$ is to have the highest number of approvals in the
  extra round.  On the other hand, suppose that $p$ is approved 
  by at least $\lfloor n/2 \rfloor +1$ voters in $(C,V)$. 
  Let $\ell_0$ be the earliest round in which $p$ receives $\lfloor n/2 \rfloor +1$ approvals,  
  let $\ell_1, \dots, \ell_q$ denote the subsequent rounds where $p$
  receives additional approval points, and let $\ell_{q+1}$ denote the extra
  round; naturally, $\ell_0 < \ell_1 < \dots < \ell_q \leq
  \ell_{q+1}$.  Now, suppose that there is a bribery $\vect$ that
  makes $p$ a winner in $\push(C,V,\vect)$. 
  If $\vect$ does not decrease the number of approvals that $p$ has, then
  $p$ wins in round $\ell_0$ in $\push(C,V,\vect)$. However, 
  $\vect$ may bribe some of the voters who approve $p$ in order to prevent some other candidate(s)
  from winning in an earlier round. If this happens, $p$ wins in 
  some round $\ell_{q'}$ with $q'\in\{0, 1, \dots, q, q+1\}$.
  We will now argue that $q'\le \beta'$. 

  To see why this is the case, suppose that under $\vect$
  we bribe $x$ voters who approve $p$ and rank her in top $\ell_0$ positions
  and $y$ voters who approve $p$ and rank her in position $\ell_1$ or lower;
  note that $x+y\le \beta'$. Let $L=\{\ell_{i_1}, \dots, \ell_{i_r}\}\subseteq \{\ell_1, \dots, \ell_q\}$ be the list 
  of rounds such that for each $\ell\in L$ it holds that $p$ receives one or more approval points in round $\ell$
  in $(C, V)$, but not in $\push(C,V,\vect)$ (i.e., in $\vect$ all voters who approve $p$ 
  and rank it in position $\ell$ are bribed not to approve $p$). By construction, we have $|L|\le y$.
  Renumber the elements of $\{\ell_1, \dots, \ell_q\}\setminus L$ as $\ell_{j_1}, \dots, \ell_{j_s}$.
  In $\push(C,V,\vect)$
  candidate $p$ has at least $\lfloor n/2 \rfloor +1-x$ approvals in round $\ell_0$,
  and it gains at least one approval point in each of the rounds $\ell_{j_1}, \dots, \ell_{j_s}$, 
  so it is approved by a strict majority of voters in round $\ell_{j_x}$.
  Since $|L|\le y$, we have $j_x\le x+y\le\beta'$.
  Hence, if $p$ is approved
  by at least $\lfloor n/2 \rfloor +1$ voters in $(C,V)$ and $\vect$ is a successful bribery, 
  then $p$ wins in
  $\push(C,V,\vect)$ in one of the rounds $\ell_0, \dots, \ell_{q'}$,
  where $q'=\min \{ \beta',q+1\}$.
  
  We now describe our algorithm.  First, the algorithm guesses the
  round $\ell$ in which $p$ wins in 
  $\push(C,V,\vect)$ and the number $\delta(p)$ of approvals that $p$
  loses until this round; in other words, we guess $(\ell,\delta(p))$
  for which $\vect$ is an $(\ell,\delta(p))$-bribery.  By our previous
  observations, there are $\beta'+1$ choices of $\ell$ and $\beta'+1$
  choices of $\delta(p)$.
  
  The algorithm then computes the set of candidates that
  have to lose at least one point as a result of $\vect$ (this
  computation depends on whether we consider classic or simplified 
  Fallback voting). Let $R$ be the set that consists of 
  these candidates as well as candidate $p$;
  we say that the candidates in $R$ are
  \emph{relevant}.  By Observation (1), if $I$ is solvable and we guessed $\ell$
  and $\delta(p)$ correctly, the set $R$
  contains at most $\beta'+1$ candidates. The algorithm
  also computes integers $\delta_{\ell-1}(c)$ and $\delta_{\ell}(c)$ for
  each $c \in R \setminus \{p\}$ such that:
  \begin{enumerate}
  \item[(3)] an $(\ell,\delta(p))$-bribery makes $p$ a winner if and
    only if each candidate $c \in R \setminus \{p\}$ loses at
    least $\delta_{\ell-1}(c)$ points until round $\ell-1$, and at
    least $\delta_{\ell}(c)$ points until round $\ell$.
  \end{enumerate}
  The procedure for computing these integers depends on whether we consider 
  classic or simplified Fallback voting. In particular, we always have $\delta_{\ell-1}(c) \leq \delta_{\ell}(c)$
  and under simplified Fallback voting
  we have $\delta_{\ell-1}(c)=\delta_{\ell}(c)$.
  We set $\delta_{\ell-1}(c)=0$ if $\ell=1$.    

  Next, the algorithm partitions the set $\{(v,t)\mid v \in V,
  1 \leq t \leq \beta'\}$ into equivalence classes. By \emph{applying}
  a pair $(v^i,t)$ we mean bribing voter $v^i$ to decrease her
  approval count from $\ell^i$ to $\ell^i-t$. We say that $(v,t)$ and
  $(v',t')$ are equivalent if for each $c_r \in R$ it holds that 
  applying $(v,t)$ has the same effect on $c_r$ as applying $(v',t')$. More formally, $(v,t)$
  and $(v',t')$ are {\em equivalent} if for each $c_r \in R$ it holds that
  (i) $(v,t)$ and $(v',t')$  decrease the number of approvals $c_r$ gets until round $\ell-1$ 
  by the same amount 
  and 
  (ii) $(v,t)$ and $(v',t')$ decrease the number of approvals $c_r$ gets until round $\ell$ 
  by the same amount.
  A pair $(v,t)$ can behave in three possible ways
  with respect to $c_r$: it can leave its approval count until
  round $\ell$ (and hence also its approval count until round $\ell-1$)
  unchanged, it can decrease by one its
  approval count in round $\ell$ (but not in earlier rounds), 
  or it can decrease by one its approval count until round $\ell-1$, 
  thereby also decreasing its approval count until round $\ell$. 
  Therefore, there are at most
  $3^{|R|} \leq 3^{\beta'+1}$ equivalence classes.  Note that applying
  some pair $(v,t)$ instead of another pair that is equivalent to it
  does not change whether a given push action is successful or not.
  
  Finally, the algorithm proceeds as follows: it guesses an
  equivalence class, picks a cheapest pair $(v^i,t)$ from this class
  that has not been applied so far, and
  applies it. Clearly, in some of the at most
  $3^{|R|}$ branches the algorithm will choose a pair that can be extended
  to an optimal bribery, if there exists one. It repeats this step
  until it reaches the bound $\beta'$ on the total approval count
  modification; this means at most $\beta'$ branchings.  By the
  arguments above, a minimum-cost solution for $I$ can be obtained by
  taking a minimum-cost bribery among all the successful briberies
  (that is, ones that ensure $p$'s victory) considered.\medskip

  \noindent
  \emph{Positive cost functions.}\quad Let us now focus on the case of
  positive cost functions, where each bribery can only increase a
  candidate's number of approvals.
    
  For each possible $\ell$ and $\delta(p)$, the algorithm tries to
  find a minimal $(\ell,\delta(p))$-bribery $\vect$ of minimum cost
  that makes $p$ a winner.  This means considering at most $|C|+1$
  possibilities for $\ell$, and, by~(2), at most $\beta'+1$
  possibilities for $\delta(p)$. Note that under positive
  cost functions a minimal successful bribery always bribes voters so
  that in each modified vote the approval count is equal to the rank of $p$.
  Consequently, the number of bribed voters in a minimal
  successful $(\ell,\delta(p))$-bribery is $\delta(p)$, and no other
  candidate receives additional approvals in round $\ell$.
  
  Having picked $\ell$ and $\delta(p)$, the algorithm tries 
  all possible ways of choosing a (multi)set of positive
  integers $\{t[1], \dots, t[{\delta(p)}]\}$ with $t[1]+ \dots + t[{\delta(p)}] = \beta'$,
  corresponding to the increase of the approval counts of the bribed
  voters.  In other words, it guesses the non-zero elements of~$\vect$
  (but not their positions in $\vect$).
  There are at most $\beta'^{\delta(p)} \leq \beta'^{\beta'}$
  possibilities.
    
  Then, the algorithm
  computes integers $\delta(c)$, $c \in C\setminus \{p\}$, such that
  \begin{enumerate}
  \item[(3')] an $(\ell,\delta(p))$-bribery makes $p$ a winner if and
    only if each candidate $c \in C \setminus \{p\}$ gains at most
    $\delta(c)$ points until round $\ell-1$ (and hence until round
    $\ell$ as well), assuming that
  $p$ does not get a majority of votes in round $\ell-1$ or earlier
  and gains exactly $\delta(p)$ points until round $\ell$.
  \end{enumerate}
  The procedure for computing $\delta(c)$ depends on whether we consider 
  simplified or classic Fallback voting, and is polynomial-time implementable 
  in either case.   
  
  The next step of the algorithm uses the color-coding technique of Alon,
  Yuster, and Zwick~\cite{alo-yus-zwi:j:colorcoding}.
  This results in a randomized algorithm with one-sided error,
  which produces a correct output with probability at least
  $2^{-\delta(p)\beta'}$; this algorithm can then be derandomized using standard
  methods.
  
  We associate a color with each bribed voter. Colors are denoted by
  integers between $1$ and $\delta(p)$; recall that $\vect$ bribes
  exactly $\delta(p)$ voters, and we have $\delta(p) \leq \beta'$.
  We construct a coloring $A:C\to 2^{\{1,\dots, \delta(p)\}}$ by assigning  
  each candidate $c \in C$ a random subset $A(c)$ of colors chosen 
  uniformly and independently.  Intuitively, for a candidate $c \in C
  \setminus\{p\}$ that obtains additional approvals as a result of the
  bribery $\vect$, the set $A(c)$ corresponds to the set of bribed
  voters in $\push(C,V,\vect)$ who grant an additional
  approval to $c$.  For candidate $p$, the set $A(p)$ corresponds to the set
  of bribed voters that grant an additional approval to $p$ until
  round $\ell-1$.  We say that $A$ is {\em valid for a candidate $c \in C
  \setminus\{p\}$} if $|A(c)| \leq \delta(c)$; it is
  {\em valid for $p$} if $|A(p)| \leq \lfloor n/2
  \rfloor-s'_{\ell-1}(p)$, where $s'_{\ell-1}(p)$ denotes the number
  of approval points $p$ gets in the original election until round $\ell-1$. 
  A coloring is {\em valid} if it is valid for all candidates in $C$.
  The concept of validity reflects the fact that we have to fulfill condition (3').
  
  Given a valid coloring of the candidates $A$, the algorithm computes the set
  of \emph{admissible colors} for each voter $v^i \in V$. A color 
  $x\in \{1, \dots, \delta(p)\}$ is {\em admissible for $v^i$} if the
  following holds:
  \begin{itemize}
  \item[(a)] $\rank(p,v^i) = \ell^i + t[x]$ and $t[x] \leq \ell-\ell^i$, i.e., in order
    to give an extra approval to $p$ in $v^i$ (in
    round $\ell$ or earlier), the approval count has to be increased by exactly
    $t[x]$;
  \item[(b)] %
    if $\rank(p,v^i) < \ell$, then $x \in A(p)$;
  \item[(c)] for each candidate $c$ with $\ell^i < \rank(c,v^i) <
    \rank(p,v^i)$, we have %
    $x \in A(c)$.
  \end{itemize}
  We say that a collection $v^{i_1}, \dots, v^{i_{\delta(p)}}$ of
  voters is \emph{proper} if for each $x$, $1 \leq x \leq \delta(p)$,
  the color $x$ is admissible for voter $v^{i_x}$.
  
  Finally, the algorithm computes a
  proper collection of voters $v^{i_1}, \dots, v^{i_{\delta(p)}}$ 
  that minimizes the cost of a bribery where we bribe each voter $v^{i_x}$ to
  increase his approval count by $t[x]$. 
  To do this, it finds a minimum-weight maximal matching in
  the bipartite graph where 
  we have the set of voters who do not approve $p$ on one
  side, $\delta(p)$ colors on the other side, 
  there is an edge from each voter to all colors   
  that are admissible for him, and the weight of each edge
  corresponds to the cost of bribing the respective voter 
  to shift his approval threshold so as to approve $p$.
  Note that a matching of size $\delta(p)$ in this graph corresponds
  to a proper collection of voters; if there is no such matching in the graph,
  then the algorithm does not output anything.

  The correctness of our algorithm is based on the following key observation: 
  if a collection of voters
  $v^{i_1}, \dots, v^{i_{\delta(p)}}$ is proper, then
  increasing the approval count of $v^{i_x}$ by $t[x]$
   for each $x=1,\dots,\delta(p)$ makes $p$ a
  winner in round $\ell$. To see this, note that condition (3')
  is satisfied if we apply such a bribery. Indeed, condition (a) ensures that
  $p$ gains the necessary $\delta(p)$ points until round
  $\ell$, condition (b) together with the definition of validity for
  $p$ ensures that $p$ does not win before round $\ell$, and condition
  (c) and the definition of validity for candidates in $C\setminus\{p\}$ that are
  affected by the bribery ensures that no candidate gains more extra approvals than it
  is allowed to.

  To complete the proof of correctness, it remains to show that if
  there exists an $(\ell,\delta(p))$-bribery $\vect$ of cost at most $B$
  that increases the approval counts of voters $v^{i_1}, \dots, v^{i_{\delta(p)}}$ 
  by $t[1], \dots, t[{\delta(p)}]$, then our algorithm outputs
  a bribery of cost at most $B$ 
  with probability at least $2^{-\delta(p)\beta'}$.  
  To see this, consider the event that the candidates affected by $\vect$ are colored ``as
  expected'', meaning that each candidate whose additional
  approvals under $\vect$ come from voters in the set $\{v^{i_x} \mid x \in X \}$ for some 
  $X \subseteq \{1, \dots, \delta(p)\}$ receives the set $X$ of
  colors during the coloring process.  For each such candidate
  this holds with probability $2^{-\delta(p)}$. Since there are at
  most $\beta'$ candidates that receive additional approvals as a
  result of the bribery $\vect$, it follows that with probability at
  least $2^{-\delta(p)\beta'}$ it holds that for each $x$, the color
  $x$ will be admissible for voter $v^{i_x}$. Whenever this holds,
  the algorithm will consider $\vect$ when searching for a cheapest proper
  collection of voters, and hence
  the output will be a (successful) bribery of cost at most $B$.
  Consequently, the algorithm indeed produces a correct output with probability at least
  $2^{-\delta(p)\beta'}\ge 2^{-\beta'^2}$.

  Let us now analyze the running time of our algorithm.  After
  choosing $\ell$, $\delta(p)$, and the integers $t[1], \dots,
  t[{\delta(p)}]$, the coloring process and the computation of 
  admissible colors for each of the voters can be implemented in linear
  time. A minimum-weight matching can be identified in polynomial
  time
  by, e.g., the Hungarian method %
  \cite{kuh:j:weighted-matching}.  The branchings in the
  beginning of the algorithm contribute a factor of $\beta'^{\beta'} |C|$
  to the running time, yielding an overall running time of
  $O(\beta'^{\beta'} |I|^{O(1)})$.

  To derandomize the algorithm, one can use
  $(\beta'|C|,\beta'^2)$-universal sets~\cite{nao-sch-sri:c:derandom};
  the resulting algorithm is still in FPT with respect to the parameter~$\beta'$.
\end{proof}

\section{Support Bribery for Single-Peaked Electorates}
\label{sec:SP}
One possible way to circumvent the hardness results of
Section~\ref{sec:support} is to study the complexity of {\sc support
  bribery} under restricted preferences.  Recent work
\cite{bet-sli-uhl:j:mon-cc,bra-bri-hem-hem:c:sp2,con:j:eliciting-singlepeaked,fal-hem-hem-rot:j:single-peaked-preferences}
shows that many hard problems in computational social choice become
easy if the voters' preferences can be assumed to be single-peaked.
In the next theorem we show that this is not the case for {\sc support
  bribery}, as this problem remains $\np$-hard (and also W[1]-hard with
parameter $\beta$) even for single-peaked electorates, for each of the
voting rules considered in this paper.

\begin{theorem}\label{thm:fallback-support-SP}
{\sc support bribery} under single-peaked preferences 
is $\np$-hard and {\em W[1]}-hard with respect to parameter $\beta$
for each of SP-AV, Fallback and simplified Fallback, 
even if $\sigma^i(k)=|k|$ for each $k$ and each $i=1, \dots, n$.
\end{theorem}

\begin{proof}
  We present a parameterized reduction from the W[1]-hard
  \textsc{multicolored clique}
  problem~\cite{fel-her-ros-via:j:multicolored-hardness}; the same
  reduction works for SP-AV, Fallback, and simplified Fallback. 
  Consider an instance of \textsc{multicolored clique} given by an integer $k$ and a
  graph $G=(\calV,\calE)$ with the vertex set $\calV=\{ \nu_1, \dots,
  \nu_{N} \}$ partitioned into $k$ independent sets $\calV^1, \dots,
  \calV^k$.  Without loss of generality we assume that $\nu_1 \in \calV^1$ and $\nu_{N} \in
  \calV^k$.
    
  We will construct an instance $I$ of support bribery with unit costs for a
  single-peaked election. We will set the budget 
  $B =2k^3-k$ and ensure that an optimal bribery has cost $B$
  if and only if $G$ contains a $k$-clique, and that there always is a
  successful bribery of cost at most $B+1$. Since $I$ 
  has unit costs, this would imply $B \leq \beta(I) \le %
  \beta'(I) \leq B+1$.

  We form an election $(C,V)$ contained in $I$ as follows.  For each $i=1,\dots, k$
  and each vertex $\nu_a \in \calV^i$, we introduce a candidate set 
  $C(\nu_a)=\{c^{j}_a \mid 1 \leq j \leq k, j \neq i \}$, and 
  we set $C_{\calV} = \bigcup_{\nu \in \calV} C(\nu)$.  
  We then set $C=C_{\calV} \cup \{p,q\} \cup D$, where $p$ is our preferred candidate and $D$
  is a set of dummies (see the next paragraph).
  
  We define a linear order $\lhd$ on the set of candidates as
  follows.  The first candidate in this order is $q$, the last one is
  $p$, and candidate $c^{i}_a$ precedes $c^{j}_b$ if either $a<b$, or
  $a=b$ but $i<j$. The dummy candidates are placed
  between candidates that are adjacent in the
  sequence $q, c^{2}_1, c^{3}_1, \dots, c^{k-2}_{N}, c^{k-1}_{N}, p$.
  Specifically, we place $2B$ dummies between $q$ and $c^{2}_1$, as
  well as between $c^{k-1}_{N}$ and $p$, and we place two dummies
  between every pair of adjacent candidates in $C_{\calV}$.  The linear
  order $\lhd$ is illustrated below ($\diamond$ signs stand for
  the dummies).
    
  \[
    q \overbrace{\diamond  \cdots \cdots \diamond}^{2 B \textrm{ dummies}} 
	\overbrace{
    c^{2}_{1} \diamond \diamond \,c^{3}_{1} \diamond \diamond  \cdots \diamond \diamond\, 
	c^{k-1}_{1} \diamond \diamond\, c^{k}_{1}}^{C(\nu_1) \textrm{ plus $2(k-2)$ dummies}}
	\cdots \cdots
	\overbrace{
	c^{1}_N \diamond \diamond\, c^{2}_N \diamond \diamond \cdots 
	\diamond \diamond\, c^{k-2}_N \diamond \diamond\, c^{k-1}_{N}}
	^{C(\nu_N) \textrm{ plus $2(k-2)$ dummies}}
	\overbrace{\diamond \cdots \cdots \diamond}^{2 B \textrm{ dummies}} p
  \]
  
  For each vertex $\nu_a \in \calV$, we introduce a voter $w_a$, 
  and for each edge $\{\nu_a, \nu_b\} \in \calE$, we introduce a voter $w_{\{a,b\}}$.  
  We let $W_{\calV}=\{w_a \mid \nu_a \in \calV \}$ and
  $W_{\calE} = \{w_{\{a,b\}} \mid \{\nu_a,\nu_b\} \in \calE \}$.  
  To define the preferences of these voters, we need additional
  notation.  For each candidate $c$ and each integer $i$, we denote by
  $\precc^i(c)$ the $i$-th candidate before $c$ in $\lhd$, and 
  we denote by $\succc^i(c)$ the $i$-th candidate after $c$ in $\lhd$.
  We sometimes write $\succc(c)$ instead of $\succc^1(c)$ and
  $\precc(c)$ instead of $\precc^1(c)$.  
  If $c$ precedes $c'$ in $\lhd$, we denote by $c \cdots c'$ 
  the sequence of candidates from $c$ to $c'$ (inclusively) with respect to $\lhd$;
  the same sequence in reverse is denoted by $c' \cdots c$.
  The preference orders and approval
  counts of voters $w_a$ and $w_{\{a,b\}}$, where
  $\nu_a \in \calV^i$, $\nu_b \in \calV^j$, $(\nu_a,\nu_b)\in\calE$, and $a<b$,
  are given below (to simplify notation, 
    we assume $i,j \notin \{1,k\}$; it is easy to modify the construction for the case
    $i=1$ or $j=k$).
  In what follows, the order of the non-approved candidates not shown in the
  preference lists can be defined in any way that results in
  preferences that are single-peaked with respect to $\lhd$.

  \begin{tabbing}
  123\=1234567\=\kill
  \> $w_a$: \> 
  $\succc(c^{k}_a) \cdots \precc^{2k^2-5k+6}(p)
	\; | \; c^{k}_a  \cdots c^{1}_a \pref  \precc^{2k^2-5k+5}(p) \cdots p$	
  \\
  \> $w_{\{a,b\}}$: \>
  $ \succc(c^{j}_a)  \cdots \precc(c^{i}_b) \pref  
	c^{j}_a \pref c^{i}_b \pref prec(c^{j}_a) \pref \succc(c^{i}_b) \; |
  $  
  \end{tabbing}

  We will now add extra votes to ensure that
  each candidate $c \in C_{\calV} \cup \{q\}$ receives $L$ approvals,
  each dummy receives fewer than $L$ approvals, while $p$ receives
  $L-k$ approvals in $(C,V)$ for some sufficiently large integer $L$.
  
  To do so, we first make sure that all candidates in $C_{\calV}$ have
  the same number of approvals.  To this end, if some candidate $c \in C_{\calV}$ has fewer
  approvals than another candidate in $C_{\calV}$, 
  we add a vote $c \cdots \precc^{B}(c^2_1) \;|$ and a vote $c \cdots
  \succc^{B}(c^{k-1}_{N}) \;|$ to $V$.  These two votes provide one extra
  approval point to each candidate in $C_{\calV} \setminus \{c\}$ and two
  extra approval points to $c$. Thus, by adding
  such pairs of votes iteratively, we can ensure that each candidate in
  $C_{\calV}$ has the same score. After equalizing the scores of all candidates
  in $C_{\calV}$ in this manner, we introduce 
  $|\calV|+|\calE|+B+1$ additional pairs of such votes 
  for each $c \in C_{\calV}$; this ensures that at this point each dummy receives strictly fewer
  approvals than candidates in $C_{\calV}$. Let the resulting score 
  of the candidates in $C_{\calV}$ be $L$; note that $L>2B+k$.  
  To obtain the required scores for $p$ and $q$, 
  we add $L-k$ voters approving $p$ only 
  and $L$ voters with preferences of the form $q \cdots \succc^{B}(q) \;|$.  
  We denote by $V_{\mathit{init}}$ the set of voters
  added to adjust the initial scores; we let 
  $V=W_{\calV} \cup W_{\calE} \cup V_{\mathit{init}}$.
  
  This completes the construction, which is clearly polynomial in size. 
  Note that the total approval score of each candidate is as required.
  Also, it is straightforward to check that the preferences of all voters are 
  single-peaked with respect to the linear order $\lhd$.
  
  It remains to show the correctness of the reduction. Note
  that the total number of voters in $V$ is at least $3L-k$, while 
  every candidate has at most $L$ approvals. 
  As $L>k$, it follows that   
  no candidate is approved by a strict majority of voters in $(C,V)$. 
  Moreover, since $L>2B+k=4k^3-k$, after any bribery
  that does not exceed the budget the score of each candidate 
  is at most $L+B < (3L-k)/2 < \lfloor |V|/2 \rfloor +1$, and hence
  no candidate is approved by a majority of voters after any such bribery.
  Thus $p$ can be made a winner under
  each of SP-AV, Fallback, and simplified Fallback if and only if $p$
  obtains the maximum number of approvals after some bribery.  Thus,
  from now on, when we speak of a score of a candidate or a
  candidate's number of points, we refer to this candidate's number of
  approvals.  To follow our arguments, the reader may find it useful
  to keep Figure~\ref{fig:SPvotes} in mind.
  
  \begin{figure}[th]
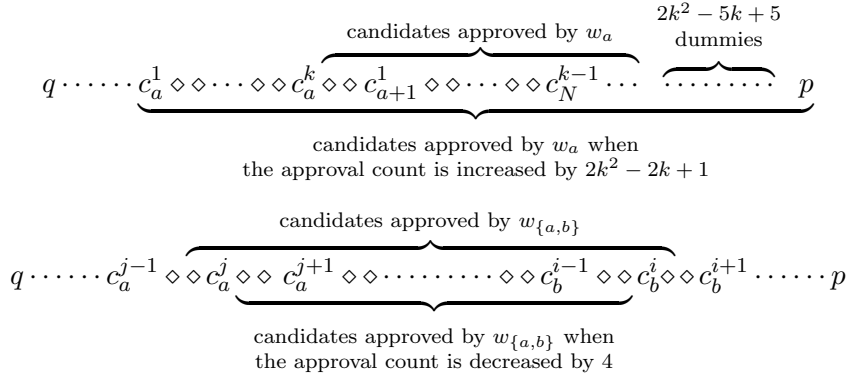
  
  $$
	q \cdots \cdots
	\underbrace{
    c^{1}_{a} \diamond \diamond 
	\cdots \diamond \diamond\, c^{k}_{a}
	\overbrace{
	\diamond \diamond
	c^{1}_{a+1} \diamond \diamond \cdots \diamond \diamond\, c^{k-1}_N  \cdots
	}^{\textrm{candidates approved by $w_a$}} 
	\overbrace{ \cdots \cdots  \cdots}^{
		\begin{array}{c}
		\textrm{\scriptsize $2k^2-5k+5$} \\[-4pt]
		\textrm{\scriptsize dummies}
		\end{array}
	} p
	}_{
		\begin{array}{c}
		\textrm{\scriptsize candidates approved by $w_a$ when} \\[-4pt]
		\textrm{\scriptsize the approval count is increased by $2k^2-2k+1$}
		\end{array}
	} 
	$$
	$$
	q \cdots 	
    \cdots c^{j-1}_{a} \diamond 
	\overbrace{
	\diamond \,c^{j}_{a} 
	\underbrace{
	\diamond \diamond \,c^{j+1}_{a} 
	\diamond \diamond 
	\cdots 
	\cdots 
	\cdots 
	\diamond \diamond 
	\,c^{i-1}_{b} \diamond \diamond}_{
		\begin{array}{c}
		\textrm{\scriptsize candidates approved by $w_{\{a,b\}}$ when} \\[-4pt]
		\textrm{\scriptsize the approval count is decreased by $4$}
		\end{array}
	}	
	c^{i}_{b} \diamond}^{
		\textrm{candidates approved by $w_{\{a,b\}}$}
	}	
	\diamond\, 
	c^{i+1}_{b} 	
	\cdots \cdots p
  $$
  \caption{\label{fig:SPvotes}
  Preferences of voters $w_a$ and $w_{\{a,b\}}$ with respect to $\lhd$.
  }
  \end{figure}
  
  First, suppose that there is a minimal successful 
  support bribery $\vect$ of cost at most $B$.
  Observe that lowering the approval counts in any of the votes
  $q \cdots \succc^{B}(q) \;|$ in order to decrease $q$'s score
  would have a cost of $B+1$. 
  Thus, $\vect$ cannot decrease $q$'s score, and hence it
  must increase $p$'s score by at least $k$ points.   
  Since $p$ is preceded by $B$ dummies in $\lhd$, 
  it follows that in $\vect$ the bribed voters form a subset
  of $W_{\calV}$.
  
  Note that increasing the approval count of a voter $w_a \in W_{\calV}$ so that 
  $p$ obtains an additional point has cost $(3k-5)+(2k^2-5k+5)+1=2k^2-2k+1$.
  Since $(k+1)(2k^2-2k+1)=2k^3-k+1>B$, at most $k$ voters 
  can be bribed in this way without exceeding the budget.
  Note also that bribing $k+1$ voters from $W_{\calV}$
  in this manner would make $p$ a winner at the cost of
  $2k^3-k+1=B+1$, which shows that $\beta(I)\le B+1$.
  
  Since $p$'s score needs to be increased by at least $k$, 
  we can conclude that $\vect$ must bribe exactly $k$ voters from the set 
  $W_{\calV}$, while increasing $p$'s score by exactly $k$.
  Let $w_{s_1}, w_{s_2}, \dots, w_{s_k}$ denote these bribed voters. 
  We are going to show that the vertex set $S=\{\nu_{s_1}, \dots, \nu_{s_k}\}$
  is a solution of the \textsc{multicolored clique} instance.
  
  Observe that when voters $w_{s_1},
  w_{s_2}, \dots, w_{s_k}$ are bribed, each of the $k(k-1)$ candidates in
  $C^*=\bigcup_{i=1}^k C(\nu_{s_i})$ receives one additional
  point. Since each of these candidates has $L$ approvals in the
  original election $(C,V)$, and the final score
  of $p$ in $\push(C,V,\vect)$ is $L$, $\vect$ must
  bribe some additional voters to lower their approval
  count so that each candidate in $C^*$ loses at least one
  point. Bribing a voter in $V_{\mathit{init}} \cup W_{\calV}$
  to decrease the score of any candidate in $C_{\calV}$
  costs more than $B$.
  Thus, to prevent the candidates in $C^*$ from beating
  $p$, $\vect$ must bribe some voters in $W_{\calE}$; let $W_{\calE}^*$
  denote the set of these voters.
  
  Bribing a voter $w_{\{a,b\}} \in W_{\calE}^*$ may:
  \begin{itemize}
  \item[(i)] decrease the score of exactly one candidate in $C_{\calV}$ at  
  a cost of $3$, or
  \item[(ii)] decrease the score of exactly two candidates in $C_{\calV}$ at  
  a cost of $4$, or
  \item[(iii)] decrease the score of $\ell$ candidates in $C_{\calV}$ for 
  some $\ell \geq 3$ at a cost of $3 (\ell-2)+4$.
  \end{itemize}
  Thus, decreasing the score of any candidate in $C^*$ has a cost of at least $2$ 
  per candidate; moreover, equality can only be achieved if case (ii) holds
  for each of the bribed voters. Hence, 
  in order to decrease the approval score by one for 
  each of the $k(k-1)$ candidates in $C^*$, the briber needs to spend
  at least $2k(k-1)$. We have argued that the briber has to spend 
  $k(2k^2-2k+1)$ on bribing voters in $W_\calV$.
  Thus, her remaining budget is $B - k(2k^2-2k+1)= 2k(k-1)$, i.e., 
  $\vect$ must bribe exactly $\binom{k}{2}$ voters from $W_{\calE}$,
  lowering the approval counts of each voter in $W_{\calE}^*$ by exactly $4$;
  moreover, both non-dummy candidates who lose points as a result of this bribery
  should be members of $C^*$.    
   
  Now, fix some vertex $\nu_x \in S$, and let $i$ be the index for
  which $\nu_x \in \calV^i$.  By the definition of $C^*$ and $S$, we
  have $C(\nu_x) \subseteq C^*$.  Therefore, for each $j$ with
  $1 \leq j \leq k$, $j \neq i$, the candidate $c^j_x \in C^*$ must 
  be among the last four approved candidates of some voter in
  $W_{\calE}^*$.  By construction of $W_{\calE}$, this voter must be
  $w_{\{x,y\}}$ for some $y$ where $\{\nu_x, \nu_y\} \in \calE$ and 
  $\nu_y \in \calV^j$.  As argued in the previous paragraph,
  this means that $c^i_y\in C^*$ and hence $\nu_y \in S$.  
  Thus, for every vertex $\nu_x$ in the set $S$, we can conclude that
  each class $\calV^j$ with $\nu_x\not\in \calV^j$ contains a vertex in
  $S \cap \calV^j$ that is adjacent to $\nu_x$.  As this holds for each
  $\nu_x \in S$, it follows that $S$ forms a clique of size $k$ in
  $G$.
  
  For the converse direction, suppose that vertices 
  $\nu_{s_1}, \dots, \nu_{s_k}$ form a clique of size $k$ in $G$.  It can be easily
  verified that lowering the approval counts of each of the voters
  $w_{\{s_i,s_j\}}$ with $1 \leq i<j \leq k$ by $4$
  and increasing the approval counts of each of the voters $w_{s_i}$ with
  $1 \leq i \leq k$ by $2k^2-2k+1$  results in a successful bribery of cost
  $4\binom{k}{2} + (2k^2-2k+1)k=B$.
\end{proof}

In Theorem~\ref{thm:fallback-support-SP} we proved that, even for
single-peaked preferences and unit costs functions, 
SP-AV-\textsc{Support Bribery} does not admit an FPT algorithm with respect 
to parameter $\beta$ unless FPT $=$ W[1].  
Naturally, this hardness result also holds for the smaller parameter $\beta'$.  
In contrast, we will now describe an algorithm that is FPT with respect to parameter $\beta'$ and
for any fixed $\eps>0$ outputs an $(1+\eps)$-approximation for this variant 
of \textsc{Support Bribery}.

\begin{theorem}\label{thm:fallback-support-SP-approx}
  For any fixed $\eps>0$, SP-AV-\textsc{support bribery} for
  single-peaked preferences can be $(1+\eps)$-approximated by an
  algorithm that is FPT with respect to $\beta'$, as long as
  $\sigma^i(k) \geq 1$ for each $k$ and each $i=1, \dots, n$.
\end{theorem}

\begin{proof}
  Fix a positive constant $\eps$.  Let $I=(C,V,\Sigma,p)$ be our input
  instance of SP-AV-\textsc{support bribery} with parameter
  $\beta'$. Just as in the proof of Theorem~\ref{thm:fallback-support-C},
  we can assume that the parameter $\beta'$ is given as part of the input; 
  otherwise, we can simply run our algorithm with increasing values 
  of the parameter starting from $1$. 
  Set $n=|V|$. For each $i=1,\dots,n$, we let $v^i$ and $\ell^i$ 
  denote the $i$-th voter in $V$ and
  her approval count. Suppose that all voters' preferences
  are single-peaked with respect to a linear order $\lhd$.
  Let $B$ be our bribery budget.

  Let $\vect=(t_1, \dots, t_n)$ be some minimal minimum-cost bribery
  such that $\beta'=\sum_{i=1}^n |t_i|$.  
  Also, let $V^+=\{v^i\in V\mid t_i>0\}$ and $V^-=\{v^i\in V\mid t_i<0\}$.

  Since $\vect$ is a minimal successful bribery, each voter $v^i \in V^+$
  ranks $p$ in position $\ell^i+t_i$. Thus,
  given a voter $v^i$ who does not approve of $p$, we will refer to the
  act of increasing $v^i$'s approval count by $\rank(p,v^i)-\ell^i$
  as \emph{buying} $v^i$.  The \emph{price} of $v^i$ is the
  cost of buying him.

  Note that since the election is single-peaked, for any voter $v^i\in V$ 
  it holds that the set of candidates
  approved by $v^i$ is a contiguous interval of the
  order $\lhd$.  Hence, the set of candidates that would receive an
  additional approval when $v^i$ is bought consists of two
  disjoint contiguous intervals of $\lhd$; one of them has
  $p$ as an endpoint (and the other one might be empty).  We write
  $B(v^i)$ and $S(v^i)$ to denote these two sets of candidates, where
  $B(v^i)$ is the one containing $p$; we refer to $B(v^i)$ and
  $S(v^i)$ as the \emph{base} and the \emph{shadow} of~$v^i$.\medskip

  \noindent
  \emph{Guessing expensive voters in $V^+$.}  Our algorithm starts by
  guessing the set of voters $V^+_1\subseteq V^+$ whose price is at
  least $\eps B$.  Since the price of each voter is at least $1$, 
  we have $|V^+_1|\le 1/\eps$. Therefore, the algorithm has
  to try at most $n^{\lfloor 1/\eps \rfloor}$ possibilities; for constant
  $\eps$, this quantity is polynomial in the input size.
  (Note also that $|V^+_1|\le |V^+|\le \beta'$, so the number of possibilities 
  to be considered at this stage can be bounded as $n^{\min\{\lfloor 1/\eps \rfloor,\beta'\}}$).
  \medskip

  \noindent
  \emph{Guessing structural properties of $\vect$.}  Next, the
  algorithm guesses the following information about $\vect$:
  \begin{enumerate}
  \item the total score $s^*$ with which $p$ wins in
    $\push(C,V,\vect)$;
  \item the size $k$ of the set $W=V^+ \setminus V^+_1$ (we 
    will refer to the voters in $W$ as $w_1,\dots, w_k$);
  \item the base $B(w_i)$ for each voter $w_i\in W$;
  \item the size $|S(w_i)|$ of the shadow for each voter $w_i\in W$.
  \end{enumerate}
  Since we have $0\le k\le \beta'$, there are at most
  $\beta'+1$ possible values of $k$.  There are also at most
  $2\beta'+1$ %
  possible choices for $s^*$ and at most $\beta'+1$ possible choices for the size of each
  shadow.  The base of each voter $w_i\in W$  is a set of candidates that is
  represented by a contiguous interval of $\lhd$ of length at most $\beta'$ whose left or
  right endpoint is $p$. This
  yields $2\beta'-1$ possible choices for each $B(w_i)$.  Thus, the
  total number of possible choices in this guessing step does not exceed
  $(2\beta'+1)(\beta'+1)^{\beta'+1} (2\beta'-1)^{\beta'}$.\medskip

  \noindent
  \emph{Color-coding step.}  To apply
  color-coding~\cite{alo-yus-zwi:j:colorcoding}, 
  we associate a color $i$ with each voter $w_i\in W$.
  Given a voter $v$ who does
  not approve $p$, we say that color $i$ is \emph{suitable} for
  $v$ if it holds that $B(v)=B(w_i)$ and $|S(v)|=|S(w_i)|$.
  The color-coding step of the algorithm assigns colors to some of the voters
  in $V$ as follows: for each voter $v \in V \setminus V^+_1$ that
  does not approve~$p$, it chooses uniformly between 
  coloring $v$ with one of the colors suitable for him and leaving $v$
  uncolored.  Recall that the algorithm does not know the voters $w_1,
  \dots, w_k$, but it has already guessed their bases and the sizes
  of their shadows, which suffices to compute the set of suitable
  colors for each voter.

  We say that the coloring is \emph{successful for a voter
  $w_i\in W$} if $w_i$ is colored with his own color $i$;
  it is {\em successful for a voter $v\in V^-$} if it leaves him uncolored.  
  A coloring is {\em successful} if it is successful for all voters in $W\cup V^-$. 
  For each of the voters in $W\cup V^-$, the probability
  that our coloring is successful for him is at least
  $1/(k+1)$, so with probability at least $(k+1)^{-\beta'}$ 
  the color-coding process results in a successful coloring. 
  Note also that if $k=\beta'$ then $V^-=\emptyset$, and we can simply color
  each voter with a suitable color (without leaving voters
  uncolored) and obtain a successful
  coloring with probability at least $\beta'^{-\beta'}$.  From now on,
  we assume that we are given a successful coloring.\medskip

  \noindent
  \emph{Guessing additional voters in $V^+$.}  Observe that if two
  voters, $v$ and $v'$, have the same base and the same shadow, then
  buying either of these voters has exactly the same effect on each
  candidate.  
  Hence, when looking for an optimal bribery, we should choose
  among such voters based on their price.   
  Based on this observation, for each color $i$, $i=1,\dots,k$, 
  we will define a set $R_i$ of {\em relevant} voters, using the following procedure.
  Let $r = \beta'^2+1$.
  We start with $R_i=\emptyset$, 
  consider the voters with color $i$ in order of non-decreasing prices,
  and place a voter $v$ into $R_i$ if and only if no other voter in $R_i$ 
  has the same shadow as $v$. We stop when $|R_i|=r$ or when we have considered 
  all voters with color $i$. Let $p_i$ denote the maximum price of a voter in $R_i$.
  By construction, we can assume without loss of generality 
  that %
  $w_i \in R_i$ or the price of $w_i$ is at least $p_i$ (in which case $|R_i|=r$).

  Now, our algorithm makes some further guesses. For each color $i$, $i=1,\dots,k$,
  it guesses whether $w_i \in R_i$; if its guess is ``yes'', it also guesses
  which voter in $R_i$ is $w_i$. For each color such a guess can have at
  most $r+1$ outcomes, so the algorithm has to try at most $(r+1)^k
  \leq (\beta'^2+1)^{\beta'}$ possibilities at this step.  Let
  $V^+_2$ be the set of voters guessed at this step.  
  Set $M=\{i\mid w_i \notin R_i\}$; we refer to colors in $M$
  as \emph{missing colors}. Note that for each $i\in M$ we have $|R_i|=r$;
  this observation will prove useful in our analysis\medskip

  \noindent
  \emph{Dynamic programming step.}  Next, the algorithm performs the following calculation
  for each candidate $c \in C \setminus \{p\}$. It computes the score that $c$ 
  would obtain if we were to buy all voters in $V^+_1$ and $V^+_2$.
  It then adds one extra point for each candidate in $B(w_i)$ for
  each missing color $i$. We denote the resulting quantity by $s^*(c)$.  
  Observe that, if we were to buy all voters
  in $V^+$, then the score of $c$ would be at least $s^*(c)$.
  Let $C^*=\{c\in C\mid s^*(c)>s^*\}$.
  Note that bribery $\vect$ buys all voters in $V^+$ and therefore
  it has to decrease the score of each candidate $c \in C^*$ 
  by at least $s^*(c)-s^*$ in order to prevent these candidates from beating
  $p$.  Clearly, $\vect$ achieves this through decreasing
  the approval counts of the voters in $V^-$.

  Instead of trying to find the set $V^-$, our algorithm simply
  computes a minimum-cost bribery $\vect^*$ that decreases the score
  of each $c \in C^*$ by at least $s^*(c)-s^*$ while bribing only
  uncolored votes. As we assume that we have a successful coloring, 
  bribing each voter in $V^-$ according to $\vect$ constitutes a feasible
  solution to this problem, and therefore the cost
  $\Sigma(\vect^*)$ does not exceed the cost of changing the approval
  count by $t_i$ in each vote $v^i \in V^-$.

  The algorithm computes $\vect^*$ by dynamic programming.
  We fix an ordering on $C^*$ and on the set of uncolored voters.
  Suppose we have $U$ uncolored voters; clearly, $U\le n$.
  Let $m^*=|C^*|$. Then for $j\in\{1,\dots, U\}$ and 
  $s_1, \dots, s_{m^*}\in\{0,\dots,\beta'\}$, we define
  $f(j,s_1, \dots, s_{m^*})$ to be the minimum cost of a bribery
  that bribes a subset of the first $j$ uncolored voters and for
  each $i=1,\dots, m^*$ decreases the number of approvals
  of the $i$-th candidate in $C^*$ by at least $s_i$. 
  We can compute $f(j,s_1, \dots, s_{m^*})$ given 
  the values of $f(j-1,s'_1, \dots, s'_{m^*})$
  for all $s'_1, \dots, s'_{m^*}\in\{0,\dots,\beta'\}$
  in time $O(|I|)$. As $|C^*|\le \beta'$, this means 
  that all values $f(j,s_1, \dots, s_{|C^*|})$
  and the bribery $\vect^*$ itself can be computed in
  $O(\beta'^{\beta'}|I|^2)$ time.\medskip

  \noindent
  \emph{Greedy phase.}  In the last step, the algorithm iteratively 
  constructs a set $V^+_3$ using the following greedy procedure. 
  We say that a voter is {\em available} if his shadow
  does not intersect the shadow of any of the voters already in $V^+_3$. 
  Initially, the algorithm sets $V^+_3=\emptyset$.
  It then considers the missing colors one by one (in any order), 
  and for each $i\in M$ it places any of the available voters from $R_i$ into $V^+_3$.
  As a final step, the algorithm picks one extra available
  vote, denoted by $v_{\mathit{extra}}$, from $R_{|M|}$. At the end of this
  procedure, the set $V^+_3$ contains exactly one vote from $R_i$ for
  each missing color $i$, $i \neq |M|$, and exactly two votes from
  $R_{|M|}$, and has the property that the shadows of the voters in $V^+_3$
  are pairwise disjoint. We will now explain why it is possible 
  to pick $|M|+1$ voters in this manner.

  Briefly, the feasibility of our greedy procedure is implied by our choice of $r$.
  In more detail, observe that, whenever the algorithm
  has to pick the next voter to add to $V^+_3$, the total size of the
  shadows of the voters already in $V^+_3$ does not exceed $\beta'$;
  this is because each voter in $R_i$ has the same
  shadow size as $w_i$, and, by definition, we have
  $\sum_{i=1}^k |S(w_i)|\le \beta'$.  Now, as the shadow of
  each voter in $R_i$ is a contiguous interval of length $|S(w_i)|$
  in the order~$\lhd$, there can be at most $|S(w_i)|$
  voters in $R_i$ whose shadows contain a certain candidate.  Thus
  there can be at most $\beta'|S(w_i)| \leq \beta'^2$ voters whose
  shadows contain any of the candidates in $\bigcup_{v \in V^+_3}
  S(v)$ and who are, hence, not available. Since $|R_i|=r$ 
  for each $i\in M$, this means that 
  for $r=\beta'^2+1$ the algorithm can always choose an
  available voter.

  We are now ready to define the output $\vect_{out}$ of the
  algorithm: this is the bribery obtained by buying each voter in
  $V^+_1 \cup V^+_2 \cup V^+_3$, and then applying the bribery
  $\vect^*$. The running time of this algorithm is
  $O(\beta'^{O(\beta')} n^{\lfloor 1/\eps \rfloor} |I|^2)$. 
  We will now prove that it outputs a successful
  bribery of cost at most $(1+\eps)B$
  with probability at least $\beta'^{-\beta'}$.

  To see that the cost of $\vect_{out}$ does not exceed $(1+\eps)B$, observe
  first that in the branch of the algorithm that performs all the
  guesses correctly, the voters in $V^+_1 \cup V^+_2$ are exactly the
  voters in $V^+ \setminus \{w_i \mid i \in M \}$. These voters are 
  bought by $\vect$, and therefore their price is present in the cost of
  $\vect$ as well. Now, consider a voter $v \in V^+_3$. If his color is $i$,
  then his price does not exceed that of $w_i$.
  Hence, the total price of the voters in
  $V^+_3 \setminus \{ v_{\mathit{extra}} \}$ does not exceed the price of
  the voters $\{w_i \mid i \in M \}$, which is also included in the cost of
  $\vect$.  The cost of the bribery $\vect^*$, which only decreases
  approval counts, is no greater than the amount spent by
  $\vect$ on bribing the voters in $V^-$.  Thus, we can conclude that
  the cost of $\vect_{out}$ is at most the cost of $\vect$ plus the
  price of $v_{\mathit{extra}}$. However, as $w_{|M|} \in V^+
  \setminus V^+_1$, we know that the price of $w_{|M|}$---and hence the
  price of $v_{\mathit{extra}}$---is less than $\eps B$.  
  This implies that the cost of $\vect_{out}$ is at most $(1+\eps)B$.

  To complete the proof, we need to show that $\vect_{out}$ succeeds in making
  $p$ a winner.  Observe that, while the bribery $\vect$ increases
  $p$'s score by $|V^+|$, the bribery $\vect_{out}$ increases $p$'s
  score by $|V^+|+1$ because of the vote $v_{\mathit{extra}}$. This
  means that the total approval score of $p$ in
  $\push(C,V,\vect_{out})$ is $s^*+1$.

  Let us fix a candidate $c \in C \setminus \{p\}$.  First, assume
  that $c$ is not contained in the shadow of any of the voters in
  $V^+_3$.  Then, by our definition of $s^*(c)$, after we buy the
  voters in $V^+_1 \cup V^+_2 \cup V^+_3$, the score of $c$ 
  is exactly $s^*(c)$.  Therefore, once we apply $\vect^*$,
  the final score of $c$ in $\push(C,V,\vect_{out})$ is at most $s^*$.
  Now, assume that $c$ is contained in the
  shadow of some voter in $V^+_3$.  Since the shadows of
  the voters in $V^+_3$ are pairwise disjoint, there is exactly 
  one voter in $V^+_3$ whose shadow contains $c$. 
  This means that, after we buy the voters in $V^+_1 \cup V^+_2 \cup V^+_3$, 
  the score of candidate $c$ is exactly $s^*(c)+1$, which implies that
  $c$'s final score in $\push(C,V,\vect_{out})$ is at most $s^*+1$.
  Hence, the score of every candidate $c \in C \setminus \{p\}$ in $\push(C,V,\vect_{out})$
  is at most $s^*+1$. It follows that $\vect_{out}$
  indeed makes $p$ a winner.

  To derandomize the algorithm, we can use standard techniques relying
  on families of perfect hash functions,
  see~\cite{alo-yus-zwi:j:colorcoding}; note that randomization only
  occurs at the color-coding step.
\end{proof}

\section{Conclusions and Future Work}
Our results show that shift bribery tends to be computationally easier
than support bribery. However, in general, the power of these campaign
management strategies is incomparable: one can construct examples of,
e.g., Fallback elections where it is impossible to make someone a
winner within a finite budget by shift bribery, but it is possible to
do so by support bribery, or vice versa. Thus, both shift bribery and
support bribery deserve to be studied in more detail.

An important contribution of this paper is the study of the
parameterized version of support bribery, where the parameter is the
total change in the approval counts.  This natural parameterization
leads to FPT algorithms for support bribery under two variants of the
Fallback rule, as well as for SP-AV, for a large class of bribery cost
functions.  Also, we presented an approximation algorithm for the case
of single-peaked preferences and unit costs that runs in FPT time with
this parameterization.  Finding other tractable parameterizations, or
more generally, identifying further tractable cases (either in the
classical, or in the parameterized sense)
is an interesting direction for future research.
The reader may wonder if it would make sense to study parameterized
complexity of shift bribery. While for the voting rules considered in
this paper shift bribery is polynomial-time solvable, for other rules
it often is $\np$-complete~\cite{elk-fal-sli:c:swap-bribery}. Very
recently, Bredereck et
al.~\cite{bre-che-fal-nic-nie:c:shift-bribery-fpt} gave a detailed
parameterized study of shift bribery for several such voting rules.

\medskip
\noindent
\textbf{Acknowledgements.}\quad  We thank the AAAI reviewers for
their comments.
Ildik\'o Schlotter was supported by the Hungarian National Research
Fund (OTKA grants no. 108383 and no. 108947).  Edith Elkind
was supported by NRF (Singapore) under Research Fellowship NRF
RF2009-08.  Piotr Faliszewski was supported in part by AGH University of
Technology Grant no.~11.11.230.124, by Polish Ministry of Science and
Higher Education grant N-N206-378637, and by Foundation for Polish
Science's program Homing/Powroty.

\bibliographystyle{plain}
\bibliography{grypiotr2006}

\appendix

\section{Destructive Support Bribery}

In this appendix, we briefly discuss destructive support bribery. In the
destructive variant of the problem, the goal is not to ensure a
preferred candidate's victory, but to prevent a despised candidate
from winning.  In contrast to our hardness results for constructive
support bribery, we can show that destructive support bribery is easy
for SP-AV, simplified Fallback voting, and Fallback voting.

\begin{theorem}
  \textsc{destructive support bribery} is in $\p$ for each of SP-AV, simplified
  Fallback voting, and Fallback voting.
\end{theorem}
\begin{proof}
  For of each of SP-AV, simplified Fallback voting, and Fallback
  voting, we use the same strategy and compute certain functions
  defined below using dynamic programming.

  Let $E = (C,V)$ be an election, where $C = \{d, c_1, \ldots,
  c_{m-1}\}$ and $V = (v^1, \ldots, v^n)$ is a collection of voters
  (each voter $v^i$ has preference order $\pref^i$ and approval count
  $\ell^i$). We are also given support bribery cost functions $\Sigma =
  (\sigma^1, \ldots, \sigma^n)$ for all voters. The outline of
  our algorithm, same for each of our voting rules, is as follows:
  \begin{enumerate}
  \item For each candidate $c \in C \setminus \{d\}$ compute the
    lowest cost of ensuring that $c$ prevents $d$ from being a winner.
  \item Output the minimum of the costs computed in the previous step.
  \end{enumerate}
  Naturally, the exact meaning of ``$c$ prevents
  $d$ from being a winner'' is different for each of our voting rules.
  For SP-AV it means
  that (a) $c$ has more approvals in total than $d$ has. For simplified
  Fallback voting it means that either 
  (a$'$) $c$ has more approvals in total than $d$ and $d$ does not win in any round, or
  (b)  $c$ wins in some round $t$ and $d$ does not win in round $t$.
  In case of Fallback voting, we need either (a$'$) or  
  (b$'$) $c$ wins in some round $t$, $d$ does not win in round $t-1$, 
  and $c$ has more $t$-approval points than $d$.

  Let us fix a candidate $c \in C \setminus \{d\}$.  To describe
  algorithms computing a lowest-cost support bribery for each of the
  above conditions, for each $k$, $i$, $j$ in $\{0,\dots, n\}$
  and for each $t$ in $\{1,\dots, m\}$ we define $f_t(k,i,j)$ to be
  the cost of a minimum-cost support bribery ensuring that exactly $i$ voters
  in $V_k = \{v^1, \ldots, v^k\}$ approve $c$ and rank her in top $t$
  positions and exactly $j$ voters in $V_k$ approve $d$ and rank
  her in top $t$ positions. It is easy to verify that
  for each $t$ and each $k, i, j$ in this range we can easily compute $f_t(k,i,j)$
  in polynomial time using standard dynamic programming techniques.

  Now, it is easy to see that the lowest cost of ensuring that $c$ has
  more approvals than $d$ (condition (a)) is exactly $\min\{ f_m(n,i,j) \mid i > j
  \}$. Similarly, the minimum cost of ensuring that condition (a$'$) holds is 
  $\min \{ f_m(n,i,j) \mid i>j, j < \left\lfloor\frac{n}{2} \right\rfloor + 1\}$.
  The lowest cost of ensuring that either $c$ wins in
  an earlier round than $d$ or $c$ wins in some round but $d$ does not (condition (b)) is
  \[
    \min\left\{ f_t(n,i,j) \mid i \geq \left\lfloor \frac{n}{2} \right\rfloor + 1, j < \left\lfloor\frac{n}{2} \right\rfloor + 1, 1 \leq t \leq m\right\}.
  \]
  To deal with the case of classic Fallback voting and compute the
  lowest cost of ensuring condition (b$'$), we have to modify our
  family of functions $f_t$ a little. For each $k$, $i$, $j$, $r$ in
  $\{0,\dots, n\}$ and each $t$ in $\{1,\dots, m\}$, let
  $f'_t(k,i,j,r)$ be the cost of a lowest-cost support bribery 
  ensuring that exactly $i$ voters in $V_k=\{v^1, \ldots, v^k\}$ approve $c$ 
  and rank her in top $t$ positions, exactly $j$ voters in $V_k$ approve $d$
  and rank her in top $t$ positions,
  and exactly $r$ voters in $V_k$ approve $d$ and rank her in top $t-1$ positions.
  Clearly, each function $f'_t$ is computable in
  polynomial time using standard dynamic programming techniques. Using
  these functions, we can compute the lowest cost of making $c$ win in
  some round where she has more approvals than $d$ does, 
  while also making sure that $d$ does not win in the previous round
  (condition (b$'$)):
  \[
    \min\left\{ f'_t(n,i,j,r) \mid i \geq \left\lfloor \frac{n}{2}
    \right\rfloor + 1, i > j, r < \left\lfloor
    \frac{n}{2}\right\rfloor +1, 1 \leq t \leq m\right\}.
  \]
  These observations show that
  destructive support bribery is in $\p$ for each of SP-AV, simplified
  Fallback voting, and Fallback voting.
\end{proof}

Not much is known about destructive shift bribery, where the briber
can ask the voters to demote the despised candidate $d$ in order to prevent her
from winning the elections; we propose algorithmic analysis
of this form of bribery as a topic for future work.
Additional motivation for the study of destructive shift bribery 
is provided by recent work that suggests some very interesting
applications of this concept~\cite{mag-riv-she-wag:c:stv-bribery,xia:margin-of-victory}.

\end{document}